\documentclass[preprint,12pt]{elsarticle}

\usepackage{amssymb}
\usepackage{amsthm}
\usepackage{mathtools}
\usepackage[dvipsnames]{xcolor}
\PassOptionsToPackage{gray}{xcolor}
\usepackage{tikz}
\usetikzlibrary{positioning,math,arrows.meta,automata}
\usepackage{xspace}
\usepackage{bm, bbm}
\usepackage{hyperref}
\usepackage{cleveref}
\usepackage{todonotes}

\usepackage{caption}
\usepackage{subcaption}
\usepackage{float}
\usepackage{wrapfig}

\newtheorem{theorem}{Theorem}[section]
\newtheorem{lemma}[theorem]{Lemma}
\newtheorem{corollary}[theorem]{Corollary}
\newtheorem{conjecture}[theorem]{Conjecture}

\newtheorem*{op}{Open Problem}

\journal{Theoretical Computer Science}

\newcommand{\torus}{\mathbb{T}}
\newcommand{\nat}{\mathbb{N}}
\newcommand{\intg}{\mathbb{Z}}
\newcommand{\rel}{\mathbb{R}}
\newcommand{\rat}{\mathbb{Q}}
\newcommand{\com}{\mathbb{C}}
\newcommand{\alg}{\overline{\rat}}
\newcommand{\ralg}{\rel \cap \alg}
\newcommand{\K}{\mathbb{K}}

\newcommand{\aut}{\mathcal{A}}

\newcommand{\Kcal}{\mathcal{K}}
\newcommand{\Lcal}{\mathcal{L}}
\newcommand{\Ocal}{\mathcal{O}}

\newcommand{\Scal}{\mathcal{S}}
\newcommand{\Tcal}{\mathcal{T}}

\newcommand{\Ical}{\mathcal{I}}

\newcommand{\eap}{effectively almost-periodic\xspace}

\newcommand{\ie}{i.e.\ }

\newcommand{\st}{\colon}
\newcommand{\cl}{\operatorname{Cl}}
\newcommand{\seq}[1]{(#1)_{n \in \mathbb{N}}}

\newcommand{\Log}{\operatorname{Log}}
\newcommand{\im}{\bm{i}}
\newcommand{\Rea}{\operatorname{Re}}
\newcommand{\Ima}{\operatorname{Im}}
\newcommand{\comp}{\mathrel{\Delta}}

\begin{document}

\begin{frontmatter}



\title{The Monadic Theory of Toric Words} 


\affiliation[irif]{organization={Universit\unexpanded{\'e} Paris Cit\unexpanded{\'e}, IRIF, CNRS},%Department and Organization
	%addressline={}, 
	city={Paris},
	postcode={F-75013}, 
	%state={},
	country={France}}
	
\affiliation[mpi]{organization={Max Planck Institute for Software Systems, Saarland~Informatics~Campus},%Department and Organization
	%addressline={}, 
	city={\\Saarbr\unexpanded{ü}cken},
	postcode={66123}, 
	%state={},
	country={Germany}}
	
\affiliation[oxford]{organization={University of Oxford, Department of Computer Science},%Department and Organization
	%addressline={Department of Computer Science, 7~Parks~Rd}, 
	city={\\Oxford},
	postcode={OX1~3QG}, 
	%state={},
	country={United Kingdom}}
	
%\affiliation[oxford]{organization={University~of~Oxford},%Department and Organization
%	addressline={Department~of Computer~Science}, 
%	city={Oxford},
%	postcode={OX1~3QG}, 
%	%state={},
%	country={United~Kingdom}}
	
\author[irif]{Val\'erie Berth\'e}      
\author[mpi]{Toghrul Karimov}
\author[mpi]{Joris Nieuwveld}
\author[mpi]{Jo\"el Ouaknine}
\author[mpi]{Mihir Vahanwala}
\author[oxford]{James Worrell}

\begin{abstract}
  For which unary predicates ${P_1,\ldots,P_m}$ is the MSO theory of
  the structure ${\langle \nat; <, P_1,\ldots,P_m \rangle}$ decidable?
  We survey the state of the art, leading us to investigate combinatorial
  properties of almost-periodic, morphic, and toric words.  In doing
  so, we show that if each $P_i$ can be generated by a toric dynamical system of a certain kind, then the attendant MSO theory is
  decidable.
  We give various applications of toric words, including the recent result of \cite{berthe2024mso} that the MSO theory of $\langle \nat; <, \{2^n \st n\in\nat\}, \{3^n \st n \in \nat\}\rangle$ is decidable.
%We also discuss recent applications of toric words to verification problems of linear recurrence sequences and linear dynamical systems.
\end{abstract}




\end{frontmatter}


\section{Introduction}
\label{sec-intro}
In 1962, B\"uchi proved in his seminal work \cite{buchi-collected-works} that the monadic second-order (MSO) theory of the structure ${\langle \nat; < \rangle}$ is decidable.
Shortly afterwards, in 1966, Elgot and Rabin \cite{elgot66_decid_undec_exten_secon_order_theor_succes} showed how to decide the MSO theory of ${\langle \nat; <, P\rangle}$ for various interesting unary predicates $P$. %initiating the study of extensions of $\langle \nat; < \rangle$ with decidable MSO theories.
On the other hand, it was known already in the 1960s that extending ${\langle \nat; <\rangle}$ with the addition or even the doubling function yields a structure with an undecidable MSO theory  \cite{robinson1958restricted,trahtenbrot1962finite}.
In this paper, we focus on the following question: which unary
predicates ${P_1,\ldots,P_m}$ can one add to $\langle \nat; < \rangle$ whilst maintaining decidability of the MSO theory?
We give an overview of the state of the art and provide some new answers.
%and discuss applications of the MSO decidability results to verification problems of linear dynamical systems.
In particular, we identify a class of predicates generated by
rotations on a torus, any number of which can be adjoined to $\langle
\nat; < \rangle$ and still preserve decidability of the attendant
monadic theory.

%for which the answer to our main question is positive. 
%We describe applications of the MSO decidability results to verification of linear dynamical systems.

By a predicate $P$ we mean a function with type $\nat\to \Sigma$, where $\Sigma$ is a finite alphabet.
When $\Sigma = \{0,1\}$, we identify $P$ with $\{n  \in \nat \st P(n) = 1\}\subseteq\nat$.
The \emph{characteristic word} of $P$ is the string $\alpha \in \Sigma^\omega$ whose $n$th letter is $P(n)$.
Let us take the primes predicate as an example, defined by ${P(n) = 1}$ if $n$ is prime and ${P(n) = 0}$ otherwise.
Recall that in a monadic second-order language we have access to the membership relation $\in$ and quantification over elements (written $Q x$ for a quantifier $Q$) as well as subsets of the universe (written $Q X$), which is $\nat$ in our case.
Consider the sentence $\psi$ given by
\begin{align*}
	\varphi(X) &\coloneqq 1 \in X \:\land\: 0,2 \notin X \:\land\: \forall x.\: x \in X \Longleftrightarrow s(s(s(x))) \in X\\
	\psi &\coloneqq \exists X\st \varphi(X) \:\land\: \forall y.\, \exists z > y\st z\in X \:\land\: P(z)
\end{align*}
where $s(\cdot)$ is the successor function defined by $s(x)=y$ if and only if
\[
x < y \:\land\: \forall z.\, x< z \Rightarrow y \le z.
\]
The formula $\varphi$ defines the subset ${\{n \st n \equiv 1 \pmod{3}\}}$ of $\nat$, and $\psi$ is the sentence ``there are infinitely many primes congruent to 1 modulo 3'', which is true.
Another example of a number-theoretic statement expressible in our setting would be the twin prime conjecture, which is given by the first-order sentence
\[
\forall x.\,\exists y > x\st P(y) \:\land\: P(s(s(y))).
\]
Unsurprisingly, the decidability of the MSO theory of the structure ${\langle \nat; <, P\rangle}$, where $P$ is the primes predicate, remains open.
Conditional decidability is known subject to Schinzel's hypothesis H, a number-theoretic conjecture which implies, in particular, the existence of infinitely many twin primes~\cite{bateman1993decidability}.
%In fact, to the best of our knowledge no natural predicate $P$ is known for which $P(n)$ can be computed effectively given $n$, and the MSO theory of ${\langle \nat; <, P\rangle}$ is undecidable.

%The original proof of decidability of the MSO theory of $\langle \nat; M \rangle$  a way to translate MSO formulas into automata.
The MSO theory of $\nat$ equipped with the order relation is intimately connected to the theory of finite automata.
The \emph{acceptance problem} for a word $\alpha \in \Sigma^\omega$, denoted $\mathsf{Acc}_\alpha$, is to decide, given a deterministic (e.g., Muller) automaton $\aut$ over $\Sigma$, whether $\aut$ accepts $\alpha$.
In order for this algorithmic problem to be well defined, we assume that the word $\alpha$ is computable;
in other words, there is a Turing machine which, upon receiving $n$ as input, prints the $n$th letter of $\alpha$.
The previously mentioned result of B\"uchi establishes that
the MSO theory of ${\langle\nat; <, P_1,\ldots,P_m\rangle}$ is decidable if and only if $\mathsf{Acc}_\alpha$ is decidable for the word ${\alpha = \alpha_1\times \cdots\times \alpha_m}$, where each $\alpha_i$ is the characteristic word of $P_i$.\footnote{The original formulation by B\"uchi was given in terms of nondeterministic B\"uchi automata. The formulations involving deterministic automata with a Muller, Rabin, or parity acceptance condition are equivalent.}
Hence our central question can be reformulated as follows: for which classes of words ${\alpha_1,\ldots,\alpha_m}$ is $\mathsf{Acc}_\alpha$ decidable?

In this work we consider the classes of \emph{almost-periodic}, \emph{morphic}, and \emph{toric} words.
Almost-periodic words were introduced by Sem\"enov in \cite{semenov84_logic_theor_one_place_funct}. 
He showed that for an \emph{effectively} almost-periodic word $\alpha$, the MSO theory of the structure ${\langle\nat; <, P_\alpha\rangle}$ is decidable, where $P_\alpha$ is the predicate whose characteristic word is $\alpha$. We discuss almost-periodic words in \Cref{sec-ap}. We then move on to morphic words (\Cref{sec-morphic}), focussing on the result of
Carton and Thomas~\cite{carton-thomas-morphic-words} that for a morphic word $\alpha$, the MSO theory of ${\langle\nat; <, P_\alpha\rangle}$ is decidable.
These two works provide answers to our main question for a single predicate, i.e., in the case of ${m=1}$.
In \Cref{sec-toric}, we introduce the class of toric words, which are codings of a rotation with respect to target sets consisting of finitely many connected components. In \Cref{toric-mso}, we give a large class $\Kcal$ of toric words such that the MSO theory of the structure ${\langle \nat; <, P_1,\ldots,P_m \rangle}$ is decidable for any number $m$ of predicates with characteristic words belonging to $\Kcal$. We also study almost periodicity and closure properties of toric words (\Cref{sec-toric-ap}), and give an account of the overlap between toric words and various other well-known families of words.
Below is a summary of how we apply the theory of toric words.

\begin{enumerate}
	\item[(a)] Sturmian words are toric. In \Cref{sec-sturmian} we use the theory of toric words to show that for Sturmian words $\alpha_1,\ldots,\alpha_m$ that satisfy a certain effectiveness assumption, the MSO theory of ${\langle \nat; <, P_{\alpha_1},\ldots,P_{\alpha_m}\rangle}$ is decidable. 
	This answers a question posed in \cite{carton-thomas-morphic-words}.
	\item[(b)] One of the central problems in symbolic dynamics is to understand the morphic words for which the associated \emph{shift space} has a representation as a \emph{geometric dynamical system} \cite{akiyama-pisot}.
	A slightly different (but similar in spirit) question is: which morphic words are toric?
	The Pisot conjecture identifies a class of morphic words for
        which a representation as a simple geometrical dynamical system is believed to exist. 
	We discuss the conjecture and how it relates morphic and toric words in \Cref{sec-pisot}.
	\item[(c)] Recently, we used the machinery of toric words to show that for a large class of predicates given by linear recurrence sequences with a single, non-repeated real dominant root, the attendant MSO theory is decidable \cite{berthe2024mso}. 
	Let integers $k_1,\ldots,k_m > 1$, and $P_i = \{k_i^n \st n \in \nat\}$ for $1 \le i \le m$.
	We showed that the MSO theory of $\langle \nat; <, P_1,P_2 \rangle$ is (unconditionally) decidable, and the MSO theory of $\langle \nat; <,P_1,\ldots,P_m\rangle$ is decidable assuming Schanuel's conjecture in transcendental number theory.
	We discuss this result in \Cref{sec-procyclic-and-sparse}.
	\item[(d)] Toric words arise naturally in the study of linear recurrence sequences.
 	In fact, specialised classes of toric words have already been used in the literature \cite{karimov-power-of-positivity,Karimov2022,karimov-popl21} to study \emph{sign patterns} of linear recurrence sequences, discussed in \Cref{sec-sign-patterns}.
	We construct sign patterns of linear recurrence sequences (LRS) that prove that the product of an almost-periodic word with a toric word that is almost-periodic need not be almost-periodic.
	\item[(e)] Finally, in \Cref{sec-lds} we give an overview of how modelling sign patterns of LRS using toric words yields decision procedures for %various subclasses of 
	the \emph{model-checking problem} for linear dynamical systems.
\end{enumerate}

%Sem\"enov and Carton and Thomas, respectively, showed that 
%\begin{enumerate}
%	\item The notion of (effectively) almost-periodic words is due to Sem\"enov.
%	Intuitively, these are words for which the gaps between consecutive occurrences of any finite word $i$ can be (effectively) bounded.
%	Sem\"nov showed that for an \eap word $\alpha$ the MSO theory of $\langle \nat; <, P_\alpha\rangle$ is decidable.
%	Unfortunately, if we take more than one almost-periodic word, then Sem\"enov's result does not apply.
%	One reason for this is that almost-periodic words are not closed under products.
%	\item A morphic word is a fixed point of a non-erasing morphism $\tau:\Sigma^* \to $
%\end{enumerate}

\section{Mathematical background}
\label{sec-mathematical-background}
By an alphabet $\Sigma$ we mean a non-empty finite set.
%, and write $\Sigma^*$ as a monoid under concatenation.
For a word $\alpha \in \Sigma^+ \cup \Sigma^\omega$ and $n \in \nat$ we let $\alpha(n)$ denote the $n$th letter of $\alpha$.
For $\alpha \in \Sigma^\omega$ we let $P_\alpha$ denote the predicate defined by $P_\alpha(n) = \alpha(n)$ for all $n$.
We write $\alpha[n,m)$ for the finite word $u = \alpha(n)\cdots\alpha(m-1)$.
Such a $u$ is called a \emph{factor} of $\alpha$.
We write $\alpha[n,\infty)$ for the infinite word $\alpha(n)\alpha(n+1)\cdots$.

Let $\alpha_i \in \Sigma_i^\omega$ for $0 \le i < L$.
The \emph{product} $\alpha_0\times \cdots\times \alpha_{L-1}$ of $\alpha_0,\ldots,\alpha_{L-1}$ is the word $\alpha$ over the product alphabet $\Sigma_0 \times \cdots \times \Sigma_{L-1}$ defined by $\alpha(n) = (\alpha_0(n),\ldots,\alpha_{L-1}(n))$.
The \emph{merge} (alternatively, the \emph{shuffling} or the \emph{interleaving}) of $\alpha_0,\ldots,\alpha_{L-1}$ is the word $\alpha$ defined by $\alpha(nL+r) = \alpha_r(n)$ for all $n \in \nat$ and $0 \le r < L$. 
Let $\Sigma_1,\Sigma_2$ be two alphabets.
A \emph{morphism} $\tau\colon \Sigma_1^*\to\Sigma_2^*$ is a map satisfying $\tau(a_1\cdots a_l) = \tau(a_1)\cdots\tau(a_l)$ for all $a_1,\ldots,a_l \in \Sigma_1$.

We write $\Log$ for the principal branch of the complex logarithm.
That is, $\Ima(\Log(z)) \in (-\pi, \pi]$ for all non-zero $z \in \com$.
For $z = {(z_1,\ldots,z_d) \in \com^d}$ and $p \ge 1$, we let $\Vert z\Vert _p$ denote the $\ell_p$ norm ${\sqrt[p]{|z_1|^p +\cdots+|z_d|^p}}$.
% and write $\Vert z\Vert $ for~$\Vert z\Vert _2$.

By a \emph{$\K$-semialgebraic} subset of $\rel^d$, where $\K \subseteq \rel$, we mean a set that can be defined by polynomial inequalities with coefficients belonging to $\K$; recall that $p(\mathbf{x}) = 0 \Longleftrightarrow p(\mathbf{x}) \ge 0 \:\land\: p(\mathbf{x})\le 0$.
A set $X \subseteq \com^d$ is \emph{$\K$-semialgebraic} if
\[
\{(x_1,y_1,\ldots,x_d,y_d) \st (x_1+y_1\im, \ldots,x_d+y_d\im) \in X\}
\] 
is a $\K$-semialgebraic subset of $\rel^{2d}$, where $\im$ denotes the imaginary unit $\sqrt{-1}$.

A sequence $\seq{u_n}$ over a ring $R$ is a \emph{linear recurrence sequence (LRS)} over $R$ if there exist $d > 0$ and ${(a_0, \ldots, a_{d-1})} \in R^d$ such that the linear recurrence
\[
u_{n+d} = a_0 u_n + \cdots + a_{d-1}u_{n+d-1}
\]
holds for all $n \in \nat$.
%\footnote{If $d=0$, then the sequence is the zero sequence.}
Here, $d$ is the \emph{order} of the linear recurrence, and the \emph{order} of an LRS $\seq{u_n}$ is the smallest number $d$ such that $\seq{u_n}$ satisfies a linear recurrence of order $d$.
An LRS $\seq{u_n}$ over $R$ of order $d$ can be written in the form $u_n = c^\top M^n s$ for some $c,s \in R^d$ and $M \in R^{d \times d}$.
If $R$ is an integral domain, then for every $p \in R[x_1,\ldots,x_d]$, $u_n = p(M^ns)$ defines an LRS over $R$.
This is a consequence of Fatou's lemma \cite[Chapter 7.2]{berstel_noncommutative}.

The most famous problem about LRS is the Skolem problem (over $\rat$): given an LRS $\seq{u_n}$ over $\rat$, decide whether there exists $n$ such that $u_n = 0$.
The Skolem problem has been open for some ninety years, counting from the seminal work \cite{skolem-sml} of Skolem, and is currently known to be decidable for LRS (over $\rat$) of order 4 or less \cite{mignotte-shorey-tijdeman-skolem,vereschagin-skolem}.
A related result is the celebrated Skolem-Mahler-Lech theorem \cite{skolem-sml,mahler-sml,lech-sml}, which asserts that the set of zeros of an LRS over a field of characteristic zero is a union of a finite set $F$ and finitely many arithmetic progressions ${a_1 + b_1\nat}, \ldots,{a_k+b_k\nat}$, where $0 \le a_i$ for all $i$.
The values of $k, a_i, b_i$ can all be effectively computed, whereas determining whether $F$ is empty is exactly the Skolem problem.
Berstel and Mignotte showed in~\cite{berstel-deux-decidable-properties-of-lrs} that for an LRS $\seq{u_n}$ there exists an effectively computable $L \ge 1$ such that for all $0 \le r  < L$, the subsequence $\seq{u_{nL+r}}$ has finitely many zeros or finitely many non-zero terms.
Consequently, if we assume existence of an oracle for the Skolem problem, then we can effectively compute all elements of $F$ in the Skolem-Mahler-Lech theorem: take $L$ subsequences and repeatedly apply the Skolem oracle to each non-zero subsequence until all zeros have been found. 

Other well-known open decision problems on LRS include the Positivity problem (given $\seq{u_n}$, decide if $u_n \ge 0$ for all $n$) and the Ultimate Positivity problem (given $\seq{u_n}$, decide if $u_n \ge 0$ for all sufficiently large~$n$).
These decision problems were already encountered in the 1970s by Salomaa and others when studying growth and related problems in formal languages \cite{rozenberg1995cornerstones,Salomaa_1978}.
The Skolem problem for LRS over $\rat$ can be reduced to the Positivity problem for LRS over $\rat$, but the latter is also, independently from the Skolem problem, hard with respect to certain open problems in Diophantine approximation~\cite{pos-low-order}.

\section{Almost-periodic words}
\label{sec-ap}

A word $\alpha \in \Sigma^\omega$ is \emph{almost-periodic} if for every finite word $u \in \Sigma^*$, there exists $R(u) \in \nat$ with the following property.
\begin{itemize}
	\item[(a)] Either $u$ does not occur in $\alpha[R(u), \infty)$, or
	\item[(b)] it occurs in every factor of $\alpha$ of length $R(u)$.
\end{itemize}
The word $\alpha$ is \emph{\eap} if (i) $\alpha(n)$ can be effectively computed for every $n$, and (ii) given $u$, we can effectively compute a value $R(u)$ with the properties above.
We represent an \eap word with two programs that compute $\alpha(n)$ on $n$ and $R(u)$ on $u$, respectively.
The word~$\alpha$ is \emph{strongly almost-periodic} if it is almost-periodic and every finite word $u$ either does not occur in $\alpha$, or occurs infinitely often.
Strongly almost-periodic words are also known as \emph{uniformly recurrent} words in the literature; see \cite{allouche_shallit_2003,Queffelec:10}.
For such words, $R(u)$ is an upper bound on the \emph{return time} of~$u$.
We will see that certain morphic words, sign patterns of linear recurrence sequences, as well as large classes of toric words are almost-periodic. 
The characteristic word $\alpha_{n!} = {01100010000\cdots}$ of the set ${\{n!\mid n \in \nat\}}$ of all factorial numbers, on the other hand, is an example of a word that is not almost-periodic.

Remarkably, for an \eap word $\alpha$ the acceptance problem $\mathsf{Acc}_\alpha$ and hence the MSO theory of the structure  $\langle\nat; <, P_\alpha \rangle$ are decidable.
We refer to this result as \textbf{Sem\"enov's theorem}.\footnote{See \cite{colcombet2011green} for a characterisation of predicates $P$ for which the MSO theory of $\langle N; <, P \rangle$ is decidable, also due to Sem\"enov.}
\begin{theorem}
	Given a deterministic automaton $\aut$ and an effectively almost-periodic word $\alpha$, it is decidable whether $\aut$ accepts $\alpha$.
\end{theorem}
See \cite{muchnik03_almos_period_sequen} for an elegant proof, showing that the sequence of states $\aut(\alpha)$ obtained when a deterministic automaton $\aut$ reads an \eap word $\alpha$ is also \eap.
It remains to determine which states occur infinitely often in $\aut(\alpha)$.
This can be done by computing $R(q)$ for every state $q$ and then checking whether $q$ occurs in ${\aut(\alpha)[R(q), 2R(q))}$.

We next give a few closure properties of almost-periodic words, which are proven in \cite{muchnik03_almos_period_sequen}.

\begin{theorem}
	Let $\alpha \in \Sigma_1^\omega$ be (effectively) almost-periodic, $\beta$ be ultimately periodic, and $\gamma$ an infinite word output by a finite-state deterministic transduces on input $\alpha$.
	Then the words $\tau(\alpha)$, $\alpha \times \beta$, and $\gamma$ are (effectively) almost-periodic.
\end{theorem}

On the other hand, by the result \cite{semenov84_logic_theor_one_place_funct} of Sem\"enov, the product of two effectively almost-periodic words need not be effectively almost-periodic.
This tells us that we cannot immediately use Sem\"enov's theorem to show decidability of the MSO theory of the structure ${\langle\nat;<,P_\alpha,P_\beta\rangle}$ for \eap words $\alpha,\beta$.
In \Cref{sec-sign-patterns}, we give explicit words $\alpha, \beta$ that are sign patterns of linear recurrences sequences and effectively almost-periodic (in fact, one of these words is toric), whereas the product $\alpha \times \beta$ is not almost-periodic.
The proof obtained in \cite{semenov84_logic_theor_one_place_funct}, in comparison, is indirect: it constructs two effectively almost-periodic words $\alpha,\beta$ that encode information about Turing machines such that the MSO theory of ${\langle \nat; <, P_{\alpha\times\beta}\rangle}$ is undecidable.
It follows that the word $\alpha\times \beta$ cannot be \emph{effectively} almost-periodic.

\section{Morphic words}
\label{sec-morphic}
By \emph{substitution} we mean a non-erasing morphism $\tau\colon \Sigma^* \to \Sigma^*$.
That is, $\tau(a) \in \Sigma^+$ for all $a \in \Sigma$. 
Let $\tau$ be a substitution and $a \in \Sigma$ be a letter such that $\tau(a) = aw$ for some $w \in \Sigma^*$.
Iterating $\tau$ on $a$, we obtain a sequence $\seq{x_n}$ of words
given by $x_0 = a$ and
$x_{n+1} = a w \tau(w) \tau^2(w)\cdots \tau^{n+1}(w)$.
For every $k, n \in \nat$, $x_{n}$ is a prefix of $x_{n+k}$.
If $|{\tau^n(a)| \to \infty}$ as ${n\to \infty}$, then $\seq{x_n}$ converges to an infinite word $\alpha \in \Sigma^\omega$ that is a fixed point of $\tau$.
Such $\alpha$ is called a \emph{substitutive} (alternatively, a
\emph{pure morphic}) word; see \cite{Queffelec:10} for an account of
the dynamics of these words.
Substitutive words are similar to and subsumed by words generated by D0L systems; the latter are obtained by iteratively applying a morphism to a word $w \in \Sigma^*$, as opposed to a single letter \cite{Kari_1997}.
We next give a few well-known examples of substitutive words.

%\definecolor{Cyan}{gray}{0.7}
%\definecolor{Rhodamine}{gray}{0.4}
%\definecolor{orange}{gray}{0.55}
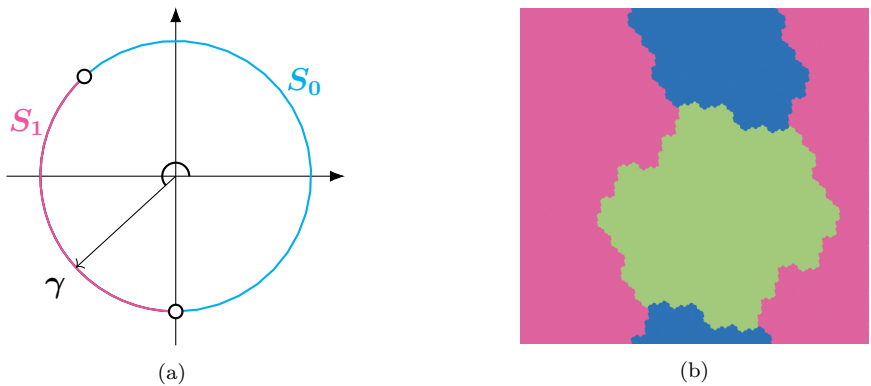
\begin{figure}[t]
	\begin{subfigure}[h]{.5\textwidth}
		\centering
		\begin{tikzpicture}[scale=0.9]
			\def\x{2.5};
			\draw[-{Latex[length=2mm]}] (0,-\x) -- (0,\x);
			\draw[-{Latex[length=2mm]}] (-\x,0) -- (\x,0);
			\draw[thick,domain=132.5:270] plot ({2*cos(\x)}, {2*sin(\x)});
			\draw[Cyan,thick,domain=-90:132.5] plot ({2*cos(\x)}, {2*sin(\x)});
			\draw[Rhodamine,thick,domain=132.5:270] plot ({2*cos(\x)}, {2*sin(\x)});
			\node[Cyan] at (1.9,1.4) {$\bm{S_0}$};
			\node[Rhodamine] at (-2.2,0.8) {$\bm{S_1}$};
			\draw[thick,domain=0:360,fill=white] plot ({0.1*cos(\x)-2*0.675}, {0.1*sin(\x)+2*0.737});
			\draw[thick,domain=0:360,fill=white] plot ({0.1*cos(\x)-2*0}, {0.1*sin(\x)-2*1});
			\draw[->] (0,0) -- (-0.737*2,-0.675*2);
			\draw[thick, domain=0:222.5] plot ({0.2*cos(\x)}, {0.2*sin(\x)});
			\node at (-0.737*2-0.3,-0.675*2-0.3) {$\bm{\gamma}$};
		\end{tikzpicture}
		\subcaption{}
	\end{subfigure}
	\begin{subfigure}[h]{.5\textwidth}
		\centering
		%\vspace*{0.5cm}
		\includegraphics*[scale=0.38]{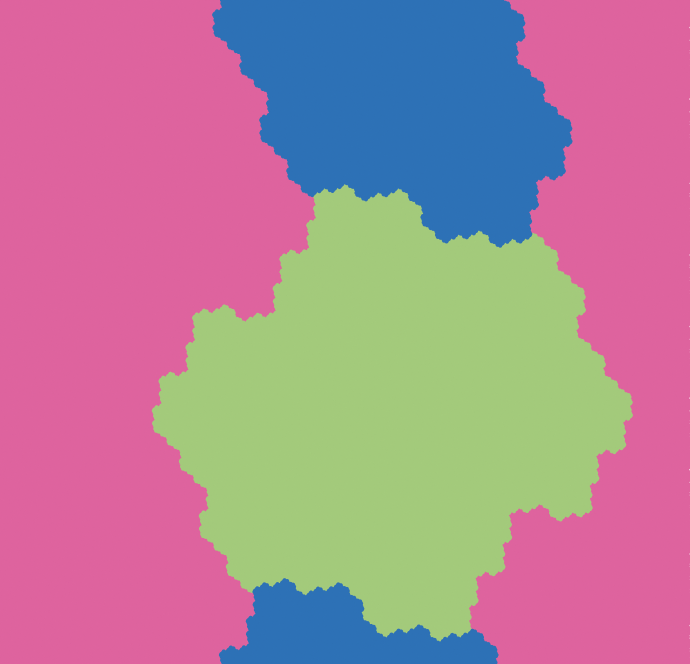}
		%\vspace*{0.75cm}
		\subcaption{}
	\end{subfigure}
	\caption{Target sets for the Fibonacci and Tribonacci words. In~(b), the pink, green, and blue sets correspond to $S_1,S_2,S_3$, respectively.}
	\label{Fibonacci and Tribonacci words} 
\end{figure}

\begin{itemize}
	\item[(a)] The Thue-Morse sequence ${0110100110\cdots}$ is generated by the substitution ${0\to01}$ and ${1\to 10}$, starting with the letter $0$.
	\item[(b)] The Fibonacci word $\alpha_F = {01001010010\cdots}$, generated by the substitution ${0\to 01}$ and ${1 \to 0}$.
	This famous sequence has many equivalent definitions, one of them as the coding of a rotation (\Cref{Fibonacci and Tribonacci words}~(a)).
	Let ${\torus = \{z \in \com \st |z| = 1\}}$.
	Let $\varphi = \frac{1+\sqrt{5}}{2} \approx 1.618$ and $\Phi = \varphi-1$ denote the golden ratio and its multiplicative inverse, respectively, and write $\gamma = e^{\im2\pi/\varphi}$.
	The long-run ratio of zeros to ones in $\alpha_F$ is equal to $1/\Phi$, and
	$\alpha_F$ is the coding of $\seq{\gamma^n}$ with respect to ${\{S_0,S_1\}}$, where $S_0,S_1$ are open interval subsets of $\torus$ with lengths $2\pi\Phi$ and $2\pi\Phi^2$, respectively.
	%Specifically, $S_0 = \{e^{\im\theta} \mid \theta \in {(-\frac \pi 2,-\frac \pi 2+\frac{2\pi}{\varphi})\}}$, and $S_1 = \{e^{\im\theta} \mid \theta \in {(-\frac \pi 2+\frac{2\pi}{\varphi}, \frac{3\pi}{2})\}}$.
	That is, for all $n \in \nat$ and $a \in \{0,1\}$, $\alpha(n)=a \Longleftrightarrow \gamma^n\in S_a$.
	We will see in \Cref{sec-families-of-words} that $\alpha_F$ is also a Sturmian word and a Pisot word.
	\item[(c)] The Tribonacci word $\alpha_T = {121312112131\cdots}$, generated by the substitution $1{\to 12}$, ${2\to 13}$, ${3\to 1}$.
	Let $\beta \approx 1.839$ be the real root of ${x^3-x^2-x-1}$ and $\Gamma = {(e^{\im2\pi/\beta},e^{\im2\pi/\beta^2})} \in \torus^2$.
	The word $\alpha_T$ has a representation as the coding of $\seq{\Gamma^n}$ with respect to three open subsets $S_1,S_2,S_3$ of~$\torus^2$ with fractal boundaries \cite{Rauzy:82}.
	For $z \in \torus$, let $f(z) = \frac{\Log(z)}{\im2\pi}+\frac{1}{2}$.
	If we identify the multiplicative group $\torus^2$ with the additive group $\rel^2/\intg^2=[0,1)^2$ via $(z_1,z_2) \to (f(z_1),f(z_2))$, the images of $S_1, S_2, S_3$ form the \emph{Rauzy fractal}. See \Cref{Fibonacci and Tribonacci words}~(b).
	\item[(d)] (Carton and Thomas~\cite{carton-thomas-morphic-words}.)
	Consider the substitution $\tau$ given by ${a \to ab}$, ${b \to ccb}$, ${c \to c}$, and let $x_n = \tau^n(a)$.
	We have that $x_1 = {ab}$, $x_2 = {abccb}$, $x_3 = {abccbccccb}$, and so on, with the fixed point $\alpha = {abc^2bc^4bc^6bc^8\cdots}$ that is not almost-periodic.
	\item[(e)]  (Salomaa, \cite{salomaa1981jewels}.) Consider the morphism $a \to aab, b \to a$. 
	The fixed point $\alpha = aabaabaaabaabaaab\cdots$ is also a Sturmian (see \Cref{sec-sturmian}) and hence a toric word \cite{berstel1994remark}.
\end{itemize}

Let $\tau$ be a substitution, and order the letters of the alphabet $\Sigma$ as ${a_1,\ldots,a_k}$. 
The matrix $M_\tau$, where $(M_\tau)_{i,j}$ is the number of occurrences of $a_j$ in $\tau(a_i)$, is called the \emph{incidence matrix} of $\tau$.
Observe that $M_\tau^n$ counts the number of occurrences of each letter in $\tau^n(a_i)$ for $1\le i \le k$.
A substitution is called \emph{primitive} if there exists $n$ such that all entries of $M_\tau^n$ are strictly positive.

The factorial word $\alpha_{n!} \in \{0,1\}^\omega$, \ie the characteristic word of the set ${\{n!\mid n \in \nat\}}$, is not substitutive.
This can be shown by observing that every fixed point of a substitution $\tau$ can be factorised as ${a\tau^0(w)\tau^1(w)\tau^2(w)\cdots}$ where ${\seq{|\tau^n(w)|}}$ grows at most exponentially. 
The blocks of zeros of $\alpha_{n!}$, however, grow super-exponentially.
Substitutive words need not be almost-periodic (see Example~(d) above), but fixed points of primitive substitutions are strongly and effectively almost-periodic \cite[Chapter 10.9]{allouche_shallit_2003}.

%Following Carton and Thomas, 
We say that a word $\beta \in \Sigma_2^\omega$ is \emph{morphic} if there exist a substitutive word $\alpha \in \Sigma_1^\omega$ and a renaming of letters ${\mu\colon \Sigma_1 \to \Sigma_2}$ such that $\beta = \mu(\alpha)$.
As an example, if we apply the morphism $\mu$ given by ${a\to 1}$, ${b\to 1}$, ${c\to 0}$ to the word $\alpha ={ abc^2bc^4bc^6\cdots}$ above, the word $\beta = \mu(\alpha)$ we obtain is the characteristic word of the squares predicate: ${\beta(n) = 1} \Longleftrightarrow {n = k^2}$ for some ${k \in \nat}$.
Carton and Thomas \cite{carton-thomas-morphic-words} showed that, in fact, for every integer $m \ge 1$ and polynomial $p \in \intg[x]$ satisfying $p(n) \ge 0$ for $n \in \nat$, the characteristic word of the set ${\{p(n)m^n\st n \in \nat\}}$ is morphic.
%Observe that if $p$ is not a constant polynomial or $m > 1$, then the characteristic word is not almost-periodic. 
Morphic words moreover subsume the class of automatic words \cite[Chapter 6.3]{allouche_shallit_2003}.

%Finally, we mention that it is decidable whether a given morphic word is almost-periodic, in which it is \eap.
%Recall that \eap words are closed under applications of a morphism $\mu$.
%Hence it suffices tos how to decide whether the substitutive word $a\tau(w)\tau(w^2)\cdots$ is almost-periodic.

\subsection{MSO decidability for morphic words}
In this section we discuss the semigroup approach used in \cite{carton-thomas-morphic-words} to show that, for a predicate~$P$ whose characteristic word is morphic, the MSO theory of ${\langle \nat; <, P\rangle}$ is decidable.
%Prior to this result, decidability was known for various special classes of morphic words.
%Let $m$ be a positive integer and $p \in \intg[x]$ be positive on $\nat$, and recall that the characteristic word $\alpha \in \{0,1\}^\omega$ of ${\{p(n)m^n\st n \in \nat\}}$ is morphic.
%The word $\alpha$ has a suffix $\beta$ of the form $10^{x_1}10^{x_2}1\cdots$ where $\seq{x_n}$ is strictly increasing, and hence $\alpha$ is not almost-periodic.
%This property, however, is a blessing in disguise thanks to the following result of Thomas \cite{thomas1978theory}, published in 1978.
%\begin{theorem}
%	\label{thomas-1978}
%	Suppose $\alpha \in \{0,1\}^\omega$  has a suffix $\beta$ in which the letter $1$ occurs infinitely often and the size of gaps between consecutive occurrences of $1$ is strictly increasing.
%	Then the MSO theory of $\langle \nat; <, P_\alpha \rangle$ is decidable.
%\end{theorem}
%This result is applicable to the characteristic word of ${\{p(n)m^n\st n \in \nat\}}$ described above, as well as the factorial word $\alpha_{n!}$.
%\Cref{thomas-1978} is proven using model-theoretic arguments.
Let $\aut$ be a deterministic automaton over an alphabet~$\Sigma$ with the set of states $Q$. 
We can associate a semigroup with $\aut$ as follows.
Two words $u_1,u_2 \in \Sigma^*$ are equivalent with respect to $\aut$, written $u_1 \equiv_\aut u_2$, if for every state $q$, there exist $R \subseteq Q$ and $t \in Q$ with the following property.
For $i \in \{1,2\}$, when $u_i$ is read in the state $q$, the run visits exactly the states in $R$ and ends in the state $t$.
Observe that $\Sigma^* / \equiv_\aut$ consists of finitely many equivalence classes.
Let $[u]_\aut$ denote the equivalence class of $u \in \Sigma^*$, noting that $u \equiv_\aut v$ implies $uw \equiv_\aut vw$ and $wu \equiv_\aut wv$ for all finite words $u,v,w$.
We define the semigroup $G_\aut = \{[u]_\aut \st u \in \Sigma^*\}$ with $[u]_\aut\cdot [v]_\aut \coloneqq [uv]_\aut$.
The semigroup $G_\aut$ associated with $\aut$ has been known since the work of B\"uchi \cite{buchi-collected-works}.

Carton and Thomas \cite{carton-thomas-morphic-words} define the class of \emph{profinitely ultimately periodic} words for which the acceptance problem is decidable.
A word $\alpha$ is profinitely ultimately periodic if it has a factorisation $\alpha  = u_0u_1u_2\cdots$ into finite words $\seq{u_n}$ such that for every morphism $\sigma\colon \Sigma^* \to G$ into a finite semigroup $G$, the sequence $\seq{\sigma(u_n)}$ is ultimately periodic.
This property is \emph{effective} if given $\sigma$, we can compute $a, b \in G^*$ such that $\sigma(\alpha) = \sigma(u_0)\sigma(u_1)\cdots = a b^\omega$.

\begin{theorem}
	If $\alpha \in \Sigma^*$ is effectively profinitely ultimately periodic, then the MSO theory of $\langle \nat; <, P_\alpha\rangle$ is decidable.
\end{theorem}
\begin{proof}
	Recall that decidability of the MSO theory is equivalent to decidability of the acceptance problem for $\alpha$: given a deterministic automaton $\aut$, decide if $\aut$ accepts $\alpha$.
	Take $\sigma$ to be the morphism that maps each $u \in \Sigma^*$ to $[u]_\aut$.
	By the assumption on $\alpha$, we can effectively compute $a, b \in {(G_\aut)}^*$ such that $\sigma(\alpha) = a b^\omega$.
	It remains to extract from $a$ and $b$ the set $S$ of states that are visited infinitely often when $\aut$ reads $\alpha$, and check $S$ against the acceptance condition of $\aut$.
\end{proof}

All morphic words are effectively profinitely ultimately periodic \cite{carton-thomas-morphic-words}.
(In fact, by a closer inspection of the Ramsey theory argument used in \cite{rabinovich2007decidability} it can be shown that all words are profinitely ultimately periodic.)
Hence the MSO theory of $\langle \nat; <, P_\alpha\rangle$ for a morphic word~$\alpha$ is decidable.
Effectively profinitely ultimately periodic words also subsume all words~$\alpha$ for which  Elgot and Rabin \cite{elgot66_decid_undec_exten_secon_order_theor_succes} showed decidability of the MSO theory of ${\langle\nat;<,P_\alpha\rangle}$ using their \emph{contraction method}.
The factorial word  $\alpha_{n!}$, for example, is an effectively profinitely ultimately periodic word that is amenable to the approach of Elgot and Rabin.
%Hence the MSO theory of $\langle \nat; <, P_{\alpha}\rangle$, where $\alpha = \alpha_{n!}$, is also decidable.
The factorisation of $\alpha_{n!}$ that yields profinite ultimate periodicity is $u_0 = 0$ and for $n \ge 1$, $u_{2n-1} = 1$ and $u_{2n} = 0^{n!-(n-1)!}$.
Rabinovich (\cite{rabinovich2007decidability}, see also \cite{rabinovich2006decidable}) showed that, in fact, the MSO theory of $\langle \nat; < ,P_1,\ldots,P_m \rangle$ is decidable \emph{if and only if} $\alpha = \alpha_1\times \cdots \times \alpha_m$, where each $\alpha_i$ is the characteristic word of $P_i$, is effectively profinitely ultimately periodic.
However, 
%once we know that the MSO theory of $\langle \nat; < ,P_1,\ldots,P_m \rangle$ is decidable, the fact that~$\alpha$ is also effectively profinitely ultimately periodic does not seem to give us anything new with respect to decidability.
%If, on the other hand, 
if we do not have any \emph{a priori} information on the decidability of the MSO theory, this characterisation 
%of \cite{rabinovich2007decidability} 
does not give us any means to determine whether the word $\alpha$ is effectively profinitely ultimately periodic or not.

Effectively profinitely ultimately periodic words are not known to be closed under products, which makes the approach of \cite{carton-thomas-morphic-words} inapplicable to the case of multiple predicates.
Let $\mathsf{Squares} = \{n^2 \st n \in \nat\}$ and $\mathsf{Cubes} = \{n^3 \st n \in \nat\}$.
As discussed above, the characteristic words of both predicates are morphic and hence the MSO theories of the structures $\langle \nat; <, \mathsf{Squares}\rangle$ and $\langle \nat; <, \mathsf{Cubes}\rangle$ are decidable.
However, decidability of the MSO theory of $\langle \nat; <, \mathsf{Squares},\mathsf{Cubes} \rangle$ is currently unknown.
Decidability of the latter theory is connected to finding the solutions of the famous Mordell equation $n^2 = m^3 + K$ for which Baker showed that when $K \ne 0$, the solutions satisfy $n,m < \exp((10^{10}|K|)^{10^4})$~\cite{baker1968contributions}. 
However, it is unclear whether Baker's result is sufficient to prove decidability of the full MSO theory.

\begin{op}
	Is the MSO theory of $\langle \nat; <, \mathsf{Squares},\mathsf{Cubes}\rangle$ decidable?
\end{op}

\section{Toric words}
\label{sec-toric}

Recall that $\torus$ is the set $\{z \in \com\st |z|=1\}$, viewed as an abelian group under multiplication.
A word $\alpha \in \Sigma^\omega$ is \emph{toric} if there exist $d > 0$, a collection $\Scal ={ \{S_a \st a \in \Sigma\}}$ of pairwise disjoint subsets of~$\torus^d$, and $\Gamma \in \torus^d$ with the following properties.
Each $S_a$ has finitely many connected components (in the Euclidean topology), and
for all $n \in \nat$ and $a \in \Sigma$, 
\[
\alpha(n) = a \Longleftrightarrow \Gamma^n \in S_a.
\]
In particular, $\Gamma^n \in \bigcup_{a \in \Sigma} S_a$ for all~$n$.
We say that $\alpha$ is generated by $(\Gamma, \Scal)$.
In the symbolic dynamics literature, $\alpha$ is referred to as the \emph{coding} of the orbit $\seq{\Gamma^n}$ with respect to the collection of sets $\Scal$.
We let $\Tcal$ denote the class of all toric words.

The purpose of the topological restriction that each $S_a$ must have finitely many connected components is to avoid the situation where every word is toric with $d = 1$.
Below we define further special subclasses of toric words that will help us better classify Sturmian words, certain morphic words, sign patterns of linear recurrence sequences, and so on.
\begin{itemize}
	\item[(a)] We let $\Tcal_O$ denote the class of toric words that are generated by $(\Gamma, \Scal)$ where each set in $\Scal$ is open in the Euclidean topology on~$\torus^d$.
	\item[(b)] The class $\Tcal_{\mathit{SA}}$ comprises all toric words generated by $(\Gamma, \Scal)$ where each set in $\Scal$ is an $\rel$-semialgebraic subset of~$\torus^d$. 
	\item[(c)] Finally, we let $\Tcal_{\mathit{SA}(\rat)}$ denote the set of all words generated by $(\Gamma, \Scal)$ such that $\Gamma \in (\torus \cap \alg)^d$, \ie $\Gamma$ has algebraic entries, and each set in $\Scal$ is $\rat$-semialgebraic.
\end{itemize}
Clearly, $\Tcal_{\mathit{SA}} \supseteq \Tcal_{\mathit{SA}(\rat)} $. 
A desirable property that the latter class has is that all operations we will need to perform on $\alpha \in \Tcal_{\mathit{SA}(\rat)}$ are effective, although $\Tcal_{\mathit{SA}(\rat)}$ is not the only subclass of $\Tcal_{\mathit{SA}}$ with this property.

We have already seen that the Tribonacci word, which is generated by the morphism ${1\to12}, {2\to 13}, {3\to 1}$ and the starting letter 1, belongs to $\Tcal_{O}$: it is generated by $(\Gamma, \Scal)$ where $\Gamma\in \torus^2$ and the sets in $\Scal$ constitute the Rauzy fractal.
We will later show that Sturmian words belong to $\Tcal_{\mathit{SA}}$, and the sign patterns of various linear recurrence sequences belong to $\Tcal_{\mathit{SA}(\rat)}$.

\subsection{Orbits in $\torus^d$}
\label{sec-orbits-in-Td}
%Toric words are codings of the trajectory of the dynamical system on the torus $\torus^d$ that evolves according to $z \to \Gamma z$ and has initial configuration $z = 1$.
In order to understand toric words, we have to understand the time steps at which the orbit $\mathcal{O}(\Gamma) \coloneqq \seq{\Gamma^n}$ of $\Gamma \in \torus^d$ visits a given subset of~$\torus^d$.
In this section we will show that unlike the discrete orbit $\mathcal{O}(\Gamma)$, its Euclidean closure $\torus_\Gamma \coloneqq \cl(\mathcal{O}(\Gamma))$ is $\rat$-semialgebraic and effectively computable under some assumptions on $\Gamma$.
Moreover, $\Ocal(\Gamma)$ visits every open subset of $\torus_\Gamma$ infinitely often.

The key to proving these results is the notion of a \emph{multiplicative relation}.
We say that $(a_1,\ldots, a_d) \in \intg^d$ is a multiplicative relation of $z = (z_1, \ldots, z_d)$,  $z \in (\com^\times)^d$ if $z_1^{a_1}\cdots z_d^{a_d} = 1$.
For such $z$, 
\[
G(z) \coloneqq \{(a_1,\ldots, a_d) \in \intg^d \mid z_1^{a_1}\cdots z_d^{a_d}=1\}
\]
is called the \emph{group of multiplicative relations} of $z$.
For all $z$, $G(z)$ is a free abelian group under addition with a basis containing at most $d$ vectors from~$\intg^d$.
If the entries of $z$ are algebraic, then such a basis can be effectively computed: by a theorem of Masser \cite{masser-mult-rel-bound}, $G(z)$ has a basis $v_1,\ldots,v_m$ of vectors satisfying $\Vert v_i\Vert_2  < B$ for all $i$, where $B$ is a bound that can be effectively computed from $z$. 
It remains to find a maximally linearly independent set of vectors of the form $a = (a_1,\ldots,a_d) \in \intg^d$ with the property that $z_1^{a_1}\cdots z_d^{a_d} = 1$ and $\Vert a\Vert_2  < B$ by enumeration. 

To describe $\torus_\Gamma$ we will employ Kronecker's theorem  in simultaneous Diophantine approximation. 
For $x,y \in \rel$, let $[\![x]\!]_y$ be the distance from $x$ to a nearest integer multiple of $y$.
Further write $[\![x]\!]$ for $[\![x]\!]_1$.
The following is a classical version of Kronecker's theorem \cite{gonek-kronecker}.
\begin{theorem}
	Let $x = (x_1,\ldots,x_d)\in \rel^d$ and $y = (y_1,\ldots,y_d) \in \rel^d$ be such that for all $b \in \intg^d$,
	\[
	b \cdot x \in \intg \Rightarrow b \cdot y \in \intg.
	\]
	For every $\epsilon>0$ there exist infinitely many values $n \in \nat$ satisfying
	\[
	\sum_{j=1}^d  [\![ nx_j - y_j  ]\!]  < \epsilon.
	\] 
\end{theorem}
Writing $X = (e^{\im 2\pi x_1}, \ldots, e^{\im 2\pi x_d})$ and $Y = (e^{\im 2\pi y_1}, \ldots, e^{\im 2\pi y_d})$, the condition that for all $b \in \intg^d$, $b \cdot x \in \intg \Rightarrow b \cdot y \in \intg$
is equivalent to $G(X) \subseteq G(Y)$. 
That is, ``every multiplicative relation of $X$ is also a multiplicative relation of $Y$''.
We can now prove the main result of this section.

\begin{lemma}
	\label{toric-Kronecker-torus}
	Let $\Gamma = (\gamma_1,\ldots,\gamma_d) \in \torus^d$.
	\begin{enumerate}
		\item[(a)] If $z \in \torus^d$ is such that $G(\Gamma) \subseteq G(z)$, then for every open $O \subset \torus_\Gamma$ containing $z$ there exist infinitely many values $n \in \nat$ such that $\Gamma^n \in O$.
		\item[(b)] $\torus_\Gamma$ is equal to
		$
		\{z \in \torus^d\st G(\Gamma) \subseteq G(z)\}, is
		$
		$\rat$-semialgebraic, and is effectively computable given a basis of $G(\Gamma)$.
	\end{enumerate}
\end{lemma}
\begin{proof}
	Consider $z = (z_1, \ldots, z_d) \in \torus^d$ with $G(\Gamma) \subseteq G(z)$.
	Define $x_j = \frac{\Log(\gamma_j)}{\im2\pi}$ and $y_j = \frac{\Log(z_j)}{\im2\pi}$ for $1 \le j \le d$. 
	We have $x_j, y_j \in (-1/2, 1/2]$.
	For all $n \in \nat$, 
	\begin{align*}
		\Vert \Gamma^n - z\Vert _1 &= \sum_{j=1}^d |\gamma_j^n-z_j| \\
		&\leq \sum_{j=1}^d |\Log(\gamma_j^n/z_j)|\\
		&= \sum_{j=1}^d [\![n \Log(\gamma_j)/\im- \Log(z_j)/\im]\!]_{2\pi} \\
		&= 2\pi \sum_{j=1}^d [\![nx_j - y_j ]\!]
	\end{align*}
	where the last equality follows from the fact that $[\![x]\!]_{2\pi} = 2\pi[\![x/(2\pi)]\!]$ for all $x \in \rel$.
	Applying Kronecker's theorem, for each $\epsilon > 0$ there exist infinitely many values $n$ such that $\Vert \Gamma^n - z\Vert _1 < \epsilon$.
	This proves (a).
	
	To prove (b), let $V = \{v_1, \ldots, v_m\}$ be a basis of $G(\Gamma)$, where for all ${1 \le k \le m}$, $v_k = {(v_{k,1}, \ldots, v_{k,d})}$.
	Since for $z= (z_1, \ldots, z_d)$,
	\[
	G(\Gamma) \subseteq G(z) \quad\Longleftrightarrow\quad \bigwedge_{k=1}^m z_1^{v_{k,1}} \cdots z_d^{v_{k,d}} = 1,
	\]
	the set $\{z \in \torus^d\st G(\Gamma) \subseteq G(z)\}$ is closed and $\rat$-semialgebraic.
	It moreover contains the orbit $\Ocal(\Gamma)$ as $G(\Gamma) \subseteq G(\Gamma^n)$ for all $n \in \nat$.
	Invoking (a), the orbit $\Ocal(\Gamma)$ is dense in $\{z \in \torus^d\st G(\Gamma) \subseteq G(z)\}$.
	Hence the latter must be exactly the closure of~$\Ocal(\Gamma)$.
\end{proof}

%We next discuss closure properties of toric words under various operations.

\subsection{Closure properties of toric words}
We now investigate closure properties of toric words under various word operations.
First we will show that unlike the class of almost-periodic words, all classes of toric words that we have defined are closed under products.
%This will be useful when later proving that every $\alpha \in \Tcal_O \cup \Tcal_{\mathit{SA}}$ is almost-periodic.

\begin{theorem}
	\label{toric-product}
	Let $\alpha_0, \ldots, \alpha_{L-1} \in \Kcal$, where $\Kcal$ is one of $\Tcal,\Tcal_O, \Tcal_{\mathit{SA}}, \Tcal_{\mathit{SA}(\mathbb{Q})}$.
	The product word $\alpha = \alpha_0 \times \cdots \times \alpha_{L-1}$ also belongs to $\Kcal$. 
\end{theorem}
\begin{proof}
	Suppose each $\alpha_i \in \Sigma_i^\omega$ and is generated by $(\Gamma_i, \{S^{(i)}_a\st a \in \Sigma_i\})$, where $\Gamma_i \in \torus^{d_i}$.
	Let $\Sigma$ be the product alphabet $\Sigma_0\times \cdots\times \Sigma_{L-1}$, noting that $\alpha \in \Sigma^\omega$.
	Further let $d= d_0 + \cdots +d_{L-1}$ and $\Gamma = \prod_{r = 0}^{L-1} \Gamma_i \in \torus^{d}$. 
	For each letter $b = (a_0,\ldots,a_{L-1}) \in \Sigma$, define $S_{b} = \prod_{r=0}^{L-1} S^{(r)}_{a_r}$.
	The word $\alpha$ is toric and generated by $(\Gamma, \{S_b \st b \in \Sigma\})$.
	It remains to observe that if every $S^{(i)}_a$ is open, or $\K$-semialgebraic for $\K = \rat$ or $\K = \rel$, then the same applies to $S_b$ for every $b \in \Sigma$.
\end{proof}

The classes of toric words we consider are also closed under applications of $k$-uniform morphisms, \ie $\tau\st \Sigma_1 \to \Sigma_2$ for which $|\tau(a)| = k$ for all $a \in \Sigma_1$.
\begin{theorem}
	Let $\alpha \in \Kcal$, where $\Kcal$ is one of $\Tcal,\Tcal_O, \Tcal_{\mathit{SA}}, \Tcal_{\mathit{SA}(\mathbb{Q})}$.
	Suppose $\alpha \in \Sigma_1^\omega$, and let $\tau\colon \Sigma_1 \to \Sigma_2$ be a $k$-uniform morphism.
	The word $\beta \coloneqq \tau(\alpha)$ also belongs to $\Kcal$.
\end{theorem}
\begin{proof}
	Suppose $\alpha$ is generated by $(\Gamma, \{S_a\st a \in \Sigma_1\})$.
	The idea is to ``slow down $\Gamma$ by a factor of $k$'' and ``add a counter modulo $k$''.
	Let $\Gamma=(\gamma_1,\ldots,\gamma_d)$, $\lambda_j = e^{\im \Log (\gamma_j)/k}$ for $1 \le j \le d$, and $\Lambda = (\lambda_1,\ldots, \lambda_d)$.
	Observe that $\gamma_j = \lambda_j^k$ for all $j$.
	Further let $\omega= e^{\im 2\pi/k}$ and $B_j = \{z \in \com \st |z - \omega^j| < 1/k\}$.
	The sets $B_0, \ldots, B_{k-1}$ are open, $\rat$-semialgebraic and pairwise disjoint.
	Moreover, $\omega^n \in B_j$ if and only if $n \equiv j \pmod{k}$.
	
	%Recall that $\tau(a)(j)$ denotes the $j$th letter of the word $\tau(a)$. 
	For a letter $b \in \Sigma_2$, define
	\[
	S_b = \bigcup_{
		\substack{a \in \Sigma_1,\:\: 0\le j < |\tau(a)|
			\\\tau(a)(j) = b}} \Lambda^{j} 
		S_a \times B_j.  
	\]
	We will show that for all $n \in \nat$ and $b \in \Sigma_2$, $\beta(n) = b$ if and only if $\Lambda_1^n \in S_b$, where $\Lambda_1 = (\lambda_1,\ldots,\lambda_d,\omega)$.
	Fix $n = qk + r$ where $0 \le r < k$.
	By construction,
	\[
	\Lambda_1^n \in S_b \Longleftrightarrow \exists a,j \st 
	\tau(a)(j) = b,\: \Lambda^n \in \Gamma^{j}S_a, \textrm{ and } \omega^n \in B_j.
	\]
	Recall that $\omega^n \in B_j$ is equivalent to $j = r$.
	Hence
	\[
	\Lambda^n \in \Gamma^{j}S_a \:\Longleftrightarrow\: \Gamma^{-j}\Lambda^n \in S_a \:\Longleftrightarrow\: \Gamma^q \in S_a \:\Longleftrightarrow\: \alpha(q) = a.
	\] 
	Above we used the fact that $\Gamma = \Lambda^{k}$.
	We have thus shown that $\Lambda_1^n \in S_b$ if and only if $\alpha(q) = a$ for some $a \in \Sigma_1$ satisfying $\tau(a)(r) = b$.
	Since $\beta(n) = \tau(\alpha(q))(r)$, it follows that $\Gamma^n \in S_b$ if and only if $\beta(n) = b$.
	That is, $\beta$ is the toric word generated by $(\Lambda_1, \{S_b\st b \in \Sigma_2\})$.
\end{proof}

\begin{corollary}
\label{toric-interleaving}
	The merge $\alpha$ of $\alpha_0,\ldots,\alpha_{L-1} \in \Kcal$, where $\Kcal$ is one of the classes of toric words as above, also belongs to $\Kcal$. 
\end{corollary}
\begin{proof}
	Suppose $\alpha_i \in \Sigma_i$.
	Let $\tau$ be the $L$-uniform morphism that maps each $(a_0,\ldots,a_{L-1}) \in \Sigma_0\times \cdots\times \Sigma_{L-1}$ to the concatenation of $a_0, \ldots, a_{L-1}$.
	Observe that $\alpha = \tau(\alpha_0 \times \cdots \times \alpha_{L-1})$.
\end{proof}

Finally, we show that our classes of toric words  are closed under taking suffixes.
This property is shared with the classes of almost-periodic words.
\begin{theorem}
	\label{toric-closure-suffixes}
	All four classes of toric words are closed under taking suffixes.
\end{theorem}
\begin{proof}
	If $\alpha$ is generated by $(\Gamma, \{S_a\st a \in \Sigma\})$, then $\alpha[N,\infty)$ is generated by $(\Gamma, \{\Gamma^{-N}S_a\st a \in \Sigma\})$.
\end{proof}

\subsection{Almost-periodicity of toric words}
\label{sec-toric-ap}

We will now show that toric words belonging to the classes $\Tcal_O$ and $\Tcal_{\mathit{SA}}$ are almost-periodic, albeit for somewhat different reasons.
The proof for the former class is topological, whereas the proof for $\Tcal_{\mathit{SA}}$ relies on the Skolem-Mahler-Lech theorem for linear recurrence sequences.
Combined with closure under products, almost periodicity of toric words will allow us to apply Sem\"enov's theorem to the problem of deciding the MSO theory of $\langle \nat; <, P_1, \ldots, P_m\rangle$, where each $P_i$ is a predicate associated with a toric word.

\begin{theorem}
	\label{toric-TO-strongly-AP}
	Every $\alpha \in \Tcal_O$ is strongly almost-periodic.
\end{theorem}
\begin{proof}
	Consider $\alpha \in \Tcal_O$ that is generated by $(\Gamma, \{S_a \st a \in\Sigma\})$ where $\Gamma = (\gamma_1,\ldots,\gamma_d) \in \torus^d$ and each $S_a$ is an open subset of $\torus^d$.
	Let $\torus_\Gamma$ denote the closure of $\seq{\Gamma^n}$, and consider a finite word $w = w(0)\cdots w(l-1) \in \Sigma^l$. 
	The latter occurs at position $n$ in $\alpha$ if and only if 
	\[
	\bigwedge_{i=0}^{l-1}\Gamma^{n+i} \in S_{w(i)}
	\]
	which is equivalent to $\Gamma^n \in S_w$ where
	\[
	S_w \coloneqq \torus_\Gamma \cap  \bigcap_{i=0}^{l-1} \Gamma^{-i}S_{w(i)}.
	\]
	Since each $S_{w(i)} \subseteq \torus^d$ is open, $S_w$ is an open subset of $\torus_\Gamma$.
	If $S_w$ is empty, then $w$ does not occur in $\alpha$.
	Suppose therefore $S_w$ is not empty.
	
	For $k \in \nat$, let $X_k = \{z \in \torus_\Gamma \st \Gamma^kz \in S_w\}$.
	Each $X_k$ is an open subset of~$\torus_\Gamma$, and since $\seq{\Gamma^n}$ visits every open subset of $\torus_\Gamma$ infinitely often, $\{X_k \st k \in \nat\}$ is an open cover of~$\torus_\Gamma$.
	By compactness of $\torus_\Gamma$, there exists $K \in \nat$ such that $\bigcup_{k=0}^K X_k$ covers $\torus_\Gamma$.
	That is, the orbit of any point in $\torus_\Gamma$ under the action of $z \to \Gamma z$ visits $S_w$ in at most $K$ steps.
	Hence for every $n \in \nat$ there exists $0 \le k \le K$ such that $\Gamma^{n+k} \in S_w$.
	Therefore, the word $w$ is guaranteed to occur in $\alpha[n, n+K+l)$ for every~$n$.
\end{proof}
\begin{corollary}
	\label{toric-TO-TSA-seap}
	Suppose $\alpha \in \Sigma^\omega$ is generated by $(\Gamma, \{S_a\st a \in \Sigma\})$ where $\Gamma \in (\torus \cap \alg)^d$ and each $S_a$ is open and $\rat$-semialgebraic..
	Then $\alpha$ is strongly and effectively almost-periodic.
\end{corollary}
\begin{proof}
	As discussed in \Cref{sec-orbits-in-Td}, we can compute the $\rat$-semialgebraic set $\torus_\Gamma$ effectively.
	Hence, given $w$, we can effectively compute a representation of $X_k$ (see the proof of \Cref{toric-TO-strongly-AP}) as a $\rat$-semialgebraic set using tools of semialgebraic geometry.
	We can then determine $K$ by checking for increasing values of $m$, starting with $m = 0$, whether $\bigcup_{k=0}^m X_k$ covers $\torus_\Gamma$.
	Hence given~$w$, we can effectively compute $K + l$ as a bound on between two consecutive occurrences of $w$ in $\alpha$.
\end{proof}
We now move on to the classes $\Tcal_{\mathit{SA}}$ and $\Tcal_{\mathit{SA}(\mathbb{Q})}$.

\begin{theorem}
	\label{toric-TSA-AP}
	Let $\alpha \in \Kcal$, where $\Kcal$ is either $\Tcal_{\mathit{SA}}$ or $\Tcal_{\mathit{SA}(\mathbb{Q})}$. 
	\begin{itemize}
		\item[(a)] There exists a suffix $\beta \coloneqq \alpha[N,\infty)$ of $\alpha$ such that $\beta \in \Tcal_O \cap \Kcal$.
		\item[(b)] The word $\alpha$ is almost-periodic.
	\end{itemize}
\end{theorem}
\begin{proof}
	Suppose $\alpha$ is generated by $((\gamma_1,\ldots,\gamma_d), \{S_a\st a \in \Sigma\})$ where each $S_a$ is a $\K$-semialgebraic subset of $\torus^d$; if $\Kcal = \Tcal_{\mathit{SA}}$, then $\K = \rel$, and $\K = \rat$ otherwise.
	Recall the definition of a semialgebraic subset of $\com^d$.
	For each letter $a$ and $z =(z_1,\ldots,z_d) \in \torus^d$, we have that $z \in S_a$ if and only if
	\[
	\bigvee_{i \in I_a} \bigwedge_{j \in J_a} p_{i,j}(\Rea(z_1),\Ima(z_1),\ldots,\Rea(z_d),\Ima(z_d)) \comp_{i,j} 0
	\]
	where each $p_{i,j}$ is a polynomial with real coefficients and ${\comp_{i,j}} \in \{\ge, >\}$.
	Define
	\[
	u^{a,i,j}_n = p_{i,j}(\Rea(\gamma_1^n),\Ima(\gamma_1^n),\ldots,\Rea(\gamma_d^n),\Ima(\gamma_d^n)).
	\]
	Observe that each $\seq{u^{a,i,j}_n}$ is a linear recurrence sequence over $\rel$.
	Applying the Skolem-Mahler-Lech theorem, for each $a,i,j$ there exist $\nu \coloneqq N_{a,i,j}$ and $\lambda \coloneqq L_{a,i,j}$ such that for each $0 \le r < \lambda$, the subsequence $\seq{u^{a,i,j}_{\nu+n\lambda + r}}$ is either identically zero or does not have any zero terms.
	Take $N = \max_{a,i,j} N_{a,i,j}$ and $L = \prod_{a,i,j} L_{a,i,j}$.
	We have that for every $a,i,j$ and $0 \le r < L$, the subsequence $\seq{v^{a,i,j,r}_n}$ of $\seq{u^{a,i,j}_n}$ given by
	\[
	v^{a,i,j,r}_n = u^{a,i,j}_{N+nL +r}
	\]
	 is either identically zero or is never zero.
	
	Consider $\beta = \alpha[N,\infty)$ and for $0 \le r < L$, define $\beta_r$  by
	\[
	\beta_r(n) \coloneqq \beta(nL+r) = \alpha(N+nL+r)
	\]
	for all $n \in \nat$.
	We will show that $\beta_r \in \Tcal_O \cap \Kcal$ for all $r$.
	Thereafter, from \Cref{toric-product} it follows that $\beta \in \Tcal_O \cap \Kcal$, proving (a).
	Invoking \Cref{toric-TO-strongly-AP}, $\beta$ is strongly almost-periodic.
	Since $\beta$ is a suffix of $\alpha$, we conclude that $\alpha$ is almost-periodic.
	
	Fix $0 \le r < L$. 
	For every $a \in \Sigma$ and $n \in \nat$ we have that  $\beta_r(n) = a$ if and only if
	\[
	\bigvee_{i \in I_a} \bigwedge_{j \in J_a} v^{a,i,j,r}_{n} \comp_{i,j} 0
	\]
	where ${\comp_{i,j}} \in \{\ge, >\}$.
	By construction of $N,L$, for each $i \in I_a$ and $j \in J_a$, the $(i,j)$th inequality above either holds for all $n$ (in case $v^{a,i,j,r}_n$ is identically zero and $\Delta_{i,j}$ is equality), or holds if and only if $v^{a,i,j,r}_{n} > 0$.
	Hence there exist $K_a \subseteq I_a$ and $M_a \subseteq J_a$ such that for all $n$, $\beta_r(n) = a$ if and only if
	\[
	\bigvee_{i \in K_a} \bigwedge_{j \in M_a} v^{a,i,j,r}_{n} > 0.
	\]
	Let $\lambda_k = \gamma_k^L$ for $1 \le k \le d$, and observe that we can write $v^{a,i,j,r}_{n} > 0$ as
	\[
	q_{a,i,j,r}(\Rea(\lambda_1^n),\Ima(\lambda_1^n), \ldots, \Rea(\lambda_d^n),\Ima(\lambda_d^n)) > 0
	\]
	for a polynomial $q_{a,i,j,r}$ with real coefficients.
	For each $a$, define $S^{(r)}_a \subseteq \torus^d$ by
	\[
	(z_1,\ldots,z_d) \in S^{(r)}_a \Longleftrightarrow \bigvee_{i \in K_a} \bigwedge_{j \in M_a} q_{a,i,j,r}(\Rea(z_1),\Ima(z_1), \ldots, \Rea(z_d),\Ima(z_d)) > 0.
	\]
	We have that $\beta_r$ is the toric word generated by $((\lambda_1,\ldots,\lambda_d), \{S^{(r)}_a\st a \in \Sigma\})$. 
	Since each $S^{(r)}_a$ is open, $\beta_r \in \Tcal_O$.
	As discussed above, it follows that $\beta$ is strongly almost-periodic and $\alpha$ is almost-periodic.
\end{proof}
\begin{corollary}
	\label{toric-skolem-eap}
	Assuming decidability of the Skolem problem for LRS over $\ralg$, every $\alpha \in \Tcal_{\mathit{SA}(\rat)}$ is \eap.
\end{corollary}
\begin{proof}
	Suppose $\alpha$ is generated by $(\Gamma, \{S_a\st a \in \Sigma\})$, where $\Gamma \in (\torus \cap \alg)^d$ and each $S_a$ is $\rat$-semialgebraic.
	In this case, in the proof of \Cref{toric-TSA-AP} each $\seq{u^{a,i,j}_n}$ is an LRS over $\ralg$.
	If we assume decidability of the Skolem problem for LRS over $\ralg$, then using the Skolem-Mahler-Lech theorem (see \Cref{sec-mathematical-background}) we can effectively compute the values of $N_{a,i,j}, L_{a,i,j}$ and hence $N, L$ in the proof above.
	We can therefore effectively compute $(\Gamma_1,\Scal_1)$ that generates the toric word $\beta =\alpha[N,\infty)$, where $\Gamma_1 \in (\torus \cap \alg)^{d}$ and each set in~$\Scal_1$ is open and $\rat$-semialgebraic.
	Invoking \Cref{toric-TO-TSA-seap}, $\beta$ is strongly and \eap.
	Hence $\alpha$ is \eap.
\end{proof}
%The following is proven in the same way as \Cref{toric-mso}, with the only difference that existence of a Skolem oracle is assumed to guarantee effective almost periodicity of $\alpha_1,\ldots,\alpha_m$.
%\begin{corollary}
%	Let $P_1,\ldots, P_m$ be predicates with respective characteristic words $\alpha_1,\ldots,\alpha_m \in \Tcal_{\mathit{SA}(\rat)}$.
%	Assuming decidability of the Skolem Problem for LRS over $\ralg$, the MSO theory of $\langle \nat; <, P_1,\ldots,P_m\rangle$ is decidable.
%\end{corollary}

%We mention that the Skolem Problem for LRS over $\ralg$ can be reduced to the Skolem Problem for LRS over $\intg$
\Cref{toric-TO-strongly-AP} tells us that words belonging to the class $\Tcal_{\mathit{SA}}$ are, in a sense, not too different from words in the class $\Tcal_O$.
In fact, we can combine words across the two classes by taking a product, while maintaining almost periodicity.

\begin{theorem}
	\label{toric-pandora}
	Let $\alpha_0, \ldots, \alpha_{L-1} \in \Tcal_O$ and $\beta_0,\ldots,\beta_{M-1} \in \Tcal_{\mathit{SA}}$.
	The word $\delta \coloneqq \prod_{i = 0}^{L-1} \alpha_i \times \prod_{j=0}^{M-1}\beta_j$ is almost-periodic.
\end{theorem}
\begin{proof}
	Let $\alpha \coloneqq  \prod_{i = 0}^{L-1} \alpha_i$ and $\beta \coloneqq \prod_{j=0}^{M-1}\beta_j$.
	The word $\delta$, up to a renaming of letters, is equal to $\alpha \times \beta$.
	By \Cref{toric-product}, $\alpha \in \Tcal_O$ and $\beta \in \Tcal_{\mathit{SA}}$.
	By \Cref{toric-TSA-AP}, there exists~$N$ such that $\beta[N,\infty) \in \Tcal_O$.
	By closure under taking suffixes (\Cref{toric-closure-suffixes}), $\alpha[N,\infty) \in \Tcal_O$.
	Applying \Cref{toric-product}, $\delta[N, \infty) = \alpha[N,\infty) \times \beta[N,\infty)$ belongs to $\Tcal_{O}$ and hence is strongly almost-periodic.
	It follows that $\delta$ is almost-periodic. 
\end{proof}

We have thus uncovered a myriad of structures with potentially decidable MSO theories:
Suppose $P_1,\ldots, P_m$ are predicates with characteristic words $\alpha_1,\ldots,\alpha_m$ that belong to $\Tcal_O \cup \Tcal_{\mathit{SA}}$.
Then the word $\alpha \coloneqq \alpha_1 \times \cdots \times \alpha_m$ is almost-periodic by \Cref{toric-pandora}.
Recall that by Sem\"enov's theorem, a sufficient condition for decidability of the MSO theory of $\langle \nat; <, P_1,\ldots,P_m\rangle$ is \emph{effective} almost periodicity of $\alpha$.
Hence the questions arises: for which toric predicates $P_1,\ldots, P_m$ is it possible to prove effective almost periodicity of the product word? 
%As a concrete example, let $\alpha_1$ be the sign pattern of the LRS $u_n = \sin (n\theta)$, where $\theta$ is not a rational multiple of $\pi$, and $\alpha_2$ be the Tribonacci word generated by the morphism $1\to 12, 2\to 23, 3\to1$ and the starting letter $1$.
%Further let $\mathsf{Sign}$ and $\mathsf{Trib}$ be the predicates having respectively $\alpha_1$ and~$\alpha_2$ as characteristic words.
%As discussed earlier, $\alpha_1 \in \Tcal_{\mathit{SA}(\rat)}$ and $\alpha_2 \in \Tcal_O$, and both words are effectively almost-periodic.
%Hence we can separately decide the MSO theories of $\langle \nat ; <, \mathsf{Sign} \rangle$ and $\langle \nat ; <, \mathsf{Trib} \rangle$, but at the moment do not have a solution to the following.
%\begin{op}
%	Is the word $\alpha \coloneqq \alpha_1 \times \alpha_2$ \eap?
%	Is the MSO theory of $\langle \nat ; <, \mathsf{Sign}, \mathsf{Trib} \rangle$ decidable?
%\end{op}
%It would be possible that the latter theory is decidable while at
%the same time $\alpha$ is not \eap.
A similar open problem is decidability of the MSO theory of $\langle \nat; < \rangle$ extended with a morphic predicate $P_1$ and a toric predicate~$P_2$.
In this case once again we can separately decide the MSO theories of $\langle \nat ; <, P_1 \rangle$ and $\langle \nat ; <, P_2 \rangle$ by \cite{carton-thomas-morphic-words} and Sem\"enov's theorem, respectively.

We conclude this section by isolating a class of toric words which we can combine while maintaining effective almost periodicity of the product word and decidability of the resulting MSO theory.
It turns out that this family of toric words is powerful enough for proving decidability of various subclasses of the model-checking problem for linear dynamical systems, discussed in \Cref{sec-lds}.

\begin{theorem}
	\label{toric-mso}
	Let $\alpha_1,\ldots,\alpha_m$ be toric words such that each $\alpha_i$ is generated by $(\Gamma_i, \Scal_i)$ where $\Gamma_i \in (\torus \cap \alg)^{d_i}$ and each $\Scal_i$ is a collection of open and $\rat$-semialgebraic sets.
	Then $\alpha$ is strongly and effectively almost-periodic.
	\begin{itemize}
		\item[(a)] The product $\alpha = \alpha_1 \times \cdots \times \alpha_m$ is \eap.
		\item[(b)] The MSO theory of the structure $\langle\nat; <, P_{\alpha_1},\ldots,P_{\alpha_m}\rangle$ is decidable.
	\end{itemize}
\end{theorem}
\begin{proof}
	Apply the construction of \Cref{toric-product} and \Cref{toric-TO-TSA-seap} to prove~(a).
	To prove (b),  recall that by B\"uchi's construction, the decision problem for the MSO theory of the structure above reduces to the acceptance problem for $\alpha$.
	The latter is decidable by \Cref{toric-TO-TSA-seap} and Sem\"enov's theorem.
\end{proof}

We can do better if we assume existence of a Skolem oracle.
%The following is proven in the same way as \Cref{toric-mso}, with the only difference that existence of a Skolem oracle is assumed to guarantee effective almost periodicity of $\alpha_1,\ldots,\alpha_m$.
\begin{theorem}
	Let $\alpha_1,\ldots,\alpha_m \in \Tcal_{\mathit{SA}(\rat)}$.
	Assuming decidability of the Skolem problem for LRS over $\ralg$, the MSO theory of $\langle \nat; <, P_{\alpha_1},\ldots,P_{\alpha_m}\rangle$ is decidable.
\end{theorem}
\begin{proof}
	%Suppose the Skolem problem is decidable for real algebraic LRS.
	Let $\alpha = \alpha_1 \times \cdots \times \alpha_m$.
	By \Cref{toric-product}, $\alpha \in \Tcal_{\mathit{SA}}$, and by \Cref{toric-skolem-eap}, $\alpha$ is \eap under the assumption that the Skolem problem is decidable for real algebraic LRS.
	It remains to invoke Sem\"enov's theorem.
\end{proof}

\section{Applications}
\label{sec-families-of-words}
In this section we discuss MSO decidability and almost periodicity properties of Sturmian words, Pisot words, sign patterns of linear recurrence sequences, certain sequences of arithmetic origin, and words arising from linear dynamical systems.

\subsection{Sturmian words}
\label{sec-sturmian}
An infinite word over the alphabet $\Sigma = \{0,1\}$ is Sturmian if the number of its distinct factors of length $n$ is equal to $n+1$ for all $n \in\nat$.
We refer the reader to \cite[Chapter 10.5]{allouche_shallit_2003} for a detailed discussion of Sturmian words.
It is known that if a word has at most $n$ distinct factors of length $n$ for some $n>0$, then it is ultimately periodic.
Hence Sturmian words have the smallest factor complexity among words that are not ultimately periodic.

Sturmian words have many equivalent characterisations, including one as a family of toric words.
For $z \in \torus$ and $x \in \rel_{>0}$, let $\Ical(z,x)$ be the open interval subset of the unit circle $\torus$ generated by starting at $z$ and rotating counter-clockwise until $ze^{\im x}$ is reached.
Further define $\Ical [z, x) \coloneqq \{z\} \cup \Ical(z, x)$ and $\Ical(z,x] \coloneqq \Ical(z, x)\cup \{ze^{\im x}\}$.
A word $\alpha$ is Sturmian if and only if there exist $\gamma \in \torus$ not a root of unity and $\xi \in \torus$ such that for all $n$, $\alpha(n)=1$ if and only if $\gamma^n \in \Ical [\xi, \theta)$, where $\theta = |\Log(\gamma)|$.
That is, a Sturmian word is the coding of $\seq{\gamma^n}$ for some $\gamma$ that is not a root of unity with respect to a partition $\{S_0, S_1\}$ of $\torus$ where $S_1$ is a semi-open interval of length exactly~$\theta$.\footnote{Note that $\theta = |\Log(\overline{\gamma})|$ and $\gamma^n \in \Ical (\xi, \theta]$ if and only if $\overline{\gamma}^n \in \Ical [\,\overline{\xi\gamma},\theta)$. Hence it suffices to only consider closed-open intervals when defining Sturmian words.}
Hence all Sturmian words belong to $\Tcal_{\mathit{SA}}$, and are almost-periodic by \Cref{toric-TSA-AP}.
In fact, they are strongly almost-periodic~\cite{allouche_shallit_2003}.

Carton and Thomas \cite{carton-thomas-morphic-words} asked:
Is the MSO theory of $\langle \nat; <, P_\alpha\rangle$, where $\alpha$ is a Sturmian word, decidable?
Call the Sturmian word with parameters $\gamma$ and $\xi$ \emph{effective} if there exists an algorithm for approximating
%we $\alpha(n)$ can be effectively computed for al $n\in \nat$ and 
$\Log(\xi)$ and $\theta \coloneqq |\Log(\gamma)|$ to arbitrary precision.
%and decide if there exists $N \in \nat$ such that $\gamma^N = \xi$.
%In case such $N$ exists (which will be unique as $\gamma$ is not a root of unity), further suppose it can be determined effectively.
We will show that such $\alpha$ is \eap and hence the MSO theory of $\langle \nat; <, P_\alpha\rangle$ is decidable.
Note that by the assumption that $\gamma$ is not a root of unity, the equation $\gamma^n = \xi$ can have at most one solution in~$n$.
Moreover, $\gamma^n = \xi e^{\im \theta}$ if and only if $\gamma^{n+1} = \xi$ or $\gamma^{n-1} = \xi$.
Hence for every effective Sturmian word $\alpha$ there exists an algorithm that computes $\alpha(n)$ given~$n$.
The algorithm simply stores the value $N$ (if any) such that $\gamma^N = \xi$, as well as the values of $\alpha(N-1),\alpha(N), \alpha(N+1)$.\footnote{Here we only show existence of the desired algorithm. If we want to write such an algorithm down, we have to first determine, if any, the value of $N$. Techniques for accomplishing this depend on the values of $\xi,\gamma$ and how they are presented.}
On $n \notin \{N-1, N, N+1\}$, it determines $\alpha(n)$ by approximating $\Log (\gamma^n)$ to sufficient precision and comparing it to approximations of $\Log(\xi)$ and $\Log(\xi e^{\im \theta})$.

\begin{theorem}
	\label{sturmian-1word}
An effective Sturmian word $\alpha$ is effectively almost-periodic.
\end{theorem}
\begin{proof}
	Suppose $\alpha$ is generated by $\gamma$ and $\xi$.
	Define $\theta$, $S_0$ and $S_1$ as above.
	As mentioned earlier, all Sturmian words are strongly almost-periodic.
	Moreover, under the assumption on $\alpha$, there exists a program that computes $\alpha(n)$ given $n$.
	Hence we have to show existence of a program that, given a finite word~$u$, determines whether $u$ occurs in $\alpha$, and in case it does, computes an upper bound on the gaps between consecutive occurrences.
	If $\gamma^N = \xi$ for some $N$, then let $M = N+2$.
	Otherwise, let $M = 0$.
	For $n \ge M$, $\gamma^n \ne \xi, \xi e^{\im\theta}$.
	That is, $\gamma^n$ does not hit the boundary of $S_0,S_1$.
	It suffices to prove effective almost periodicity of $\beta \coloneqq \alpha[M, \infty)$.
	As in the proof of \Cref{toric-TO-strongly-AP}, a word $w=w(0)\cdots w(l-1)$ occurs at a position $n \ge M$ in $\alpha$ if and only if $\gamma^n \in S_w$, where 
	\[
	S_w =  \bigcap_{i=0}^{l-1} \gamma^{-i}S_{w(i)}
	\]
	and each $S_{w(i)}$ is the open interval $\Ical(\xi, \xi e^{\im\theta})$ if $w(i)=1$ and $S_{w(i)} = \Ical(\xi e^{\im\theta}, \xi)$ otherwise.
	Since $\gamma$ is not a root of unity, no two distinct intervals $\gamma^{-i}S_{w(i)}$ and $\gamma^{-j}S_{w(j)}$ share an endpoint.
	Hence by approximating $\Log(z)$ to sufficient precision for every endpoint $z$ of $\gamma^{-i}S_{w(i)}$ for $0 \le i < l$, we can decide whether $S_w$ is empty.
	If $S(w) = \emptyset$, then $w$ does not occur in $\beta$.
	% and $\alpha$ since $\alpha$ is strongly almost-periodic.
	If $S_w \ne \emptyset$, then we can compute, using the approximate positions of the endpoints,
	%of $\gamma^{-i}S_{w(i)}$ for $0 \le i < l$, 
	an open semialgebraic interval subset $J$ of $\torus$ that is contained in $S_w$.
	Similarly to the proofs of \Cref{toric-TO-strongly-AP,toric-TO-TSA-seap}, let $K$ be such that $\bigcup_{i=0}^K \gamma^{-i}J$ covers $\torus$; such $K$ can be computed using a trial-and-error method and tools of semialgebraic geometry.
	Thus for every $m \in \nat$ there exists $n \in [m, m+K]$ such that $\gamma^n \in J$, which implies $\gamma^n \in S_w$.
	It follows that for every $m \in \nat$ the word $w$ occurs in $\beta[m, m+K+l)$.
\end{proof}

What about decidability of the MSO theory of $\langle \nat;<, P_{\alpha_1},\ldots,P_{\alpha_m}\rangle$, where each $\alpha_i$ is Sturmian? 
Suppose each $\alpha_i$ is an effective Sturmian word with parameters $\gamma_i, \xi_i$ and $\theta_i = |\Log(\gamma_i)|$.
Suppose further that $\gamma_1,\ldots,\gamma_m$ are multiplicatively independent. Importantly, under this assumption, $\torus_\Gamma = \torus^d$ for $\Gamma = (\gamma_1,\ldots,\gamma_d)$.
% and a representation of the $\rat$-semialgebraic set $\torus_\Gamma$.
\begin{theorem}
	Under the assumptions above, $\alpha \coloneqq \alpha_1 \times \cdots \times \alpha_m$ is \eap and hence the MSO theory of $\langle \nat; <, P_{\alpha_1},\ldots,P_{\alpha_m}\rangle$ is decidable.
\end{theorem}
\noindent\emph{Proof sketch.}
Let $\Sigma = \{0,1\}^m$ and $M$ be such that for all $n \ge M$ and $1 \le j \le m$, $\gamma_j^n \ne \xi_j$ and $\gamma_j^n \ne \xi_j e^{\im\theta_j}$.
For each $a \in \Sigma$, there exists $S_a \subset \torus^m$ that is a product of open interval subsets of $\torus$ (henceforth called a \emph{box}) such that for all $n \in \nat$, $\alpha(n) = a$ if and only if $\Gamma^n \in S_a$.
Let $w \in \Sigma^l$.
For $n \ge M$, the word $w$ occurs at the position $n$ in $\alpha$ if and only if $\Gamma^n \in S_w$, where $S_w = \bigcap_{i=0}^{l-1} \Gamma^{-i}S_{w(i)}$ and each $S_{w(i)}$ is of the form $\prod_{j=1}^m T^{(i)}_j$ for open intervals ${T^{(i)}_1, \ldots,T^{(i)}_{m} \subset \torus}$.
Therefore,
\[
S_w  =\prod_{j=1}^m\: \bigcap_{i=0}^{l-1}\gamma_j^{-i} T^{(i)}_j
\] 
itself is an open box.
As argued in the proof of \Cref{sturmian-1word}, using the oracles for approximating $\Log(\gamma_i), \Log(\xi_i)$ to sufficient precision we can decide whether each $\bigcap_{i=0}^{l-1}\gamma_j^{-i} T^{(i)}_j$ is empty.
In case $S_w$ is non-empty, we compute an open semialgebraic box $J$ such that $J \subset S_w$.
It remains to bound the return time of $\seq{\Gamma^n}$ in $J$ by 
computing $K$ such that $\bigcup_{i=0}^K\Gamma^{-i}J$ covers~$\torus_\Gamma$, which is the whole of $\torus^d$ by the multiplicative independence assumption.
Since $J$ is $\rat$-semialgebraic, such $K$ can be computed effectively by trial-and-error.
In the end, for every $m \ge M$, the word $w$ occurs in $\alpha[m,m+K+l)$.
\qed

We mention that for a \emph{characteristic Sturmian word} $\alpha$ generated by a quadratic irrational (see \cite[Chap.~9]{allouche_shallit_2003}), the \emph{first-order} theory of the structure $\langle \nat; <, +, n \mapsto \alpha(n)\rangle$ is decidable by the automata-theoretic methods of Hieronymi \emph{et al}.\ 
\cite{hieronymi2018ostrowski,hieronymi2024decidability}.
Note that in this theory we have access to addition, but not to second-order quantification.
Because the continued fraction expansions of quadratic irrationals are ultimately periodic, a word $\alpha$ as above is, in fact, morphic \cite[Chap.~9]{allouche_shallit_2003}.

\subsection{Pisot words}
\label{sec-pisot}
%Which morphic words have a geometric representation as a coding of a trajectory of a dynamical system?
%This is one of the most fundamental questions in symbolic dynamics.
We now discuss a class of morphic words called \emph{Pisot words} and the related Pisot conjecture.
The conjecture identifies a class of morphic words that are expected to have, in a specific sense, a toric representation.

A \emph{Pisot–Vijayaraghavan number}, also called a \emph{Pisot number}, is a real algebraic integer greater than $1$ whose Galois conjugates all have absolute value less than 1.
A \emph{Pisot} substitution  $\tau \colon \Sigma^* \to \Sigma^*$ has the property that the incidence matrix $M_\tau$ of $\tau$ has a single real dominant eigenvalue that is a Pisot number.
A morphic word generated by a Pisot substitution is called a \emph{Pisot word}.
The Fibonacci and Tribonacci words we encountered are both Pisot words that also belong to $\Tcal_O$.
The Fibonacci word is the coding of a rotation with respect to two interval subsets of $\torus$, whereas the Tribonacci word is the coding of $\seq{\Gamma^n}$, where $\Gamma = (e^{\im \frac{2\pi}{x}}, e^{\im \frac{2\pi}{x^2}})$ and $x \approx 1.839$ is the largest root of the polynomial $x^3-x^2-x-1$, with respect to $\Scal =\{S_1,S_2,S_3\}$ with fractal boundaries (see \Cref{sec-morphic}).

To state the Pisot conjecture, we first need a few definitions.
The \emph{language} $\Lcal(\alpha)$ of $\alpha \in \Sigma^\omega$ is the set of all factors of $\alpha$.
Recall that a substitution $\tau\colon\Sigma^*\to \Sigma^*$ is primitive if there exists $k \in \nat$ such that starting from any letter $a$, $\tau^k(a)$ contains all possible letters.
Further recall that a fixed point of a primitive substitution is strongly and effectively almost-periodic.
A substitution $\tau$ is \emph{unimodular} if $\det(M_\tau) = \pm 1$.
Finally, $\tau$ is \emph{irreducible} if the characteristic polynomial of $M_\tau$ is irreducible.
Now we are ready to state the Pisot conjecture.
\begin{conjecture}[Pisot conjecture]
    If $\alpha$ is a fixed point of a unimodular, primitive and irreducible Pisot substitution over a $k$-letter alphabet, then there exists a word $\beta$ with the following properties.
\begin{itemize}
	\item[(a)] $\Lcal(\beta) = \Lcal(\alpha)$, and
	\item[(b)] $\beta$ is the toric word generated by some $(\Gamma, \Scal)$ where $\Gamma \in \torus^{k-1}$ and each set in $\Scal$ is open.
\end{itemize}
\end{conjecture}
Statement~(b) implies $\beta \in \Tcal_{O}$.
Note that by~(a), the word $\beta$ is also strongly and effectively almost-periodic.
The Pisot conjecture is widely believed to be true but has only been proven for $k = 2$; see \cite{akiyama-pisot} for a detailed account.

\subsection{Procyclic and sparse predicates}
\label{sec-procyclic-and-sparse}
The results of this section were recently obtained in \cite{berthe2024mso} using a combination of tools from number theory, automata theory, and symbolic dynamics.
Let $P = \{f(n) \st n \in \nat\}$, where $f \st \nat \to \nat$ is strictly increasing.
We say that $P$ is \emph{procyclic} if given $m \ge 1$, we can effectively compute $N, p$ such that $f(n+p) \equiv f(n) \pmod{m}$ for all $n \ge N$.
Now consider $P_1, P_2$ given by $P_i = \{f_i(n) \st n \in \nat\}$, where each $f_i$ is strictly increasing.
The pair of predicates $P_1,P_2$ is said to be \emph{effectively sparse} if for every $K\in \nat$, the set $\{(n_1,n_2) \st |f_1(n_1)-f_2(n_2)| \le K\}$ is finite and can be effectively computed.

For predicates $P_1,\ldots,P_m$ with respective characteristic words $\alpha_1,\ldots,\alpha_m$, we write $\operatorname{Ord}(P_1,\ldots,P_m)$ for the word over $(\{0,1\}^m \setminus (0,\ldots,0))^\omega$, called the \emph{order word}, obtained by deleting all occurrences of the letter $(0,\ldots,0)$ from $\alpha_1\times \cdots\times \alpha_m$.
We have the following.

\begin{theorem}
	\label{thm::char-to-order-word}
	Let $P_1,\ldots,P_m$ be predicates with respective characteristic words $\alpha_1,\ldots,\alpha_m$, $\alpha = \alpha_1 \times \cdots \times \alpha_m$, and $\beta = \operatorname{Ord}(P_1,\ldots,P_m)$.
	Suppose each $P_i$ is procyclic and the pair $P_i, P_j$ is effectively sparse for every $i \ne j$.
	Then $\mathsf{Acc}_\alpha$ reduces to $\mathsf{Acc}_\beta$.
\end{theorem}

This result is the first step in the proof of decidability of the MSO theory of $\langle \nat; <, \{2^n \st n \in \nat\}, \{3^n \st n \in \nat\}\rangle$.
To state our decidability result in full,  consider linear recurrence sequences
\[
u^{(i)}_n = c_i\rho_i^n + \sum_{k=1}^{K_i} p_{i,k}(n)\lambda_{i,k}^n
\]
over $\intg$ for $1\le i \le m$ with the following properties.
For all $i,j, k$,
\begin{itemize}
	\item $\lambda_{i,k} \in \alg$, $c_i, \rho_i \in \ralg$, $p_{i,k} \in \alg[x]$, 
	\item $c_i > 0$, $\rho_i > 1$, $|\lambda_{i,k}| < \rho_i$, and
	\item $c_i \rho_i^n = c_j\rho_j^n$ has finitely many solutions when $i \ne j$.
\end{itemize}
Write $P_i$ for $\{u^{(i)}_n \st n \in \nat\} \cap \nat$. 
We have the following.

\begin{theorem}
	\label{thm-lics24}
	The MSO  theory of $\langle \nat; P_1,\ldots,P_m\rangle$ is decidable assuming Schanuel's conjecture.
	The decidability is unconditional if either of the following holds:
	\begin{itemize}
		\item $1/\Log(\rho_1),\ldots,1/\Log(\rho_m)$ are linearly independent over $\rat$;
		\item Every triple of $\rho_1,\ldots,\rho_m$ is multiplicatively dependent, and $\rho_1,\ldots,\rho_m$ are pairwise multiplicatively independent.
	\end{itemize}
\end{theorem}
By \Cref{thm-lics24} the MSO  of $\langle \nat; <, 2^{\nat}, 3^{\nat},
6^{\nat}, 12^{\nat}\rangle$ is decidable, where we write $k^\nat$ to denote $\{k^n \st n \in \nat\}$.
The idea of the proof is to first reduce to the order word using \Cref{thm::char-to-order-word}.
It turns out that the order word obtained from predicates of the form $k^\nat$ belongs to the class of \emph{billiard words}, which are almost-periodic (in fact, uniformly recurrent) and belong to $\Tcal_O$.
If the second condition in \Cref{thm-lics24} does not hold, Schanuel's conjecture is required for computing bounds on the window function. 
In contrast to our decidability result, Hieronymi and Schulz have recently shown that the \emph{first-order} theory of $\nat$ equipped with addition and the predicates $2^\nat, 3^\nat$ is undecidable \cite{hieronymi022cobham}.

\subsection{Sign patterns of linear recurrence sequences}
\label{sec-sign-patterns}
The \emph{sign pattern} of a real-valued LRS $\seq{u_n}$ is the word $\alpha \in \{+, 0, -\}^\omega$ such that $\alpha(n)$ is defined by $\operatorname{sign}(u_n)$ for all $n \in \nat$.
The Skolem, Positivity and Ultimate Positivity problems introduced in \Cref{sec-mathematical-background} are all decision problems about such sign patterns.
We will see that sign patterns of LRS can have distinctive combinations of toricity and almost periodicity properties.

We start with \emph{simple} (also known as \emph{diagonalisable}) sequences.
%Recall that an LRS over $\ralg$ of order $d$ can be defined in the form $u_n = c^\top M^ns$, where $c,s \in (\ralg)^d$ and $M \in (\ralg)^{d \times d}$.
An LRS $\seq{u_n}$ over $\alg$ is called simple if it can be expressed in the form $u_n = c^\top M^ns$ where $c,s \in \alg^d$ and $M \in \alg^{d \times d}$ is diagonalisable.
Using a deep result \cite{evertse} of Evertse on the sums of $S$-units, we can show that the sign pattern $\alpha$ of a simple LRS $\seq{u_n}$ has a suffix that belongs to $\Tcal_{O}$.

\begin{theorem}[Theorem 11 in \cite{karimov-power-of-positivity}]
\label{simple-lrs-sign}
	Let $\seq{u_n}$ be a simple LRS over $\ralg$ with the sign pattern $\alpha \in \{+,0,-\}^\omega$.
	\begin{itemize}
		\item[(a)] There exist integers $d, N$, open semialgebraic subsets $S_+, S_0, S_-$ of $\torus^d$, and $\Gamma \in (\torus \cap \alg)^d$ such that $\alpha[N,\infty) \in \Tcal_O \cap \Tcal_{\mathit{SA}(\rat)}$ and is generated by $(\Gamma, \{S_+,S_0,S_-\})$.
		\item[(b)] The value of $N$ and representations of
                  $S_+,S_0,S_-$ can be effectively computed assuming
                  decidability of the Positivity problem for simple LRS over~$\rat$.
	\end{itemize}
\end{theorem}
Sign patterns of non-simple LRS, however, do not have such properties.
We next give an example of a sign pattern of a non-simple LRS that is almost-periodic but provably does not belong to $\Tcal_{O}$ nor to $\Tcal_{\mathit{SA}}$.
Let $\gamma = 0.6+0.8\im  \in \torus \cap \alg$ and $\theta =\Log(\gamma)/\im$, noting that $\gamma$ is not a root of unity. 
Consider the linear recurrence sequences $u_n = \sin (n\theta)$ and  $v_n = n\sin (n\theta) - 7\cos (n\theta)$. 
Write $\alpha, \beta \in \{+,0,-\}^\omega$ for their sign patterns, respectively.
%We will use the word~$\beta$ to also show that the product of an almost-periodic word with the sign pattern of $\sin(n\theta)$ need not be almost-periodic.
%Note that $\Tcal_{O} \cap \Tcal_{\mathit{SA}(\rat)}$, 

\begin{figure}[t]
	\begin{subfigure}[h]{.5\textwidth}
		\centering
		\begin{tikzpicture}[scale=1]
			\def\x{3.1};
			\draw [Cyan,thick,domain=180:360] plot ({2*cos(\x)}, {2*sin(\x)});
			\draw [Rhodamine,thick,domain=0:180] plot ({2*cos(\x)}, {2*sin(\x)});
			\node[Rhodamine] at (-0.3,2.4) {$\bm{S_+}$};
			\node[Cyan] at (0.35,-2.4) {$\bm{S_-}$};
			\node at (1.35,1.8) {$\gamma$};
			\draw [domain=0:55] plot ({0.3*cos(\x)}, {0.3*sin(\x)});
			\node at (0.45,0.22) {$\theta$};
			\draw[->] (0,0) -- (0.6*2,0.8*2);
			\draw [white,thick,domain=0:360,fill=white] plot ({2*(0.04*cos(\x)-1)}, {2*(0.02*sin(\x))});
			\draw[-{Latex[length=2mm]}] (0,-\x) -- (0,\x);
			\draw[-{Latex[length=2mm]}] (-\x,0) -- (\x,0);
			\draw [orange,thick,domain=0:360,fill=orange] plot ({2*(0.04*cos(\x)+1)}, {2*(0.04*sin(\x))});
			\draw [black,thick,domain=0:360,fill=white] plot ({2*(0.04*cos(\x)-1)}, {2*(0.04*sin(\x))});
			\node[orange] at (2.3,0.3) {$\bm{S_0}$};
		\end{tikzpicture}
		\subcaption{}
	\end{subfigure}
	\begin{subfigure}[h]{.5\textwidth}
		\centering
		\begin{tikzpicture}[scale=1]
			\def\x{3.1};
			\draw[-{Latex[length=2mm]}] (-\x,0) -- (\x,0);
			\draw[-{Latex[length=2mm]}] (0,-\x) -- (0,\x);
			\draw [Cyan,thick,domain=195:375] plot ({2*cos(\x)}, {2*sin(\x)});
			\draw [Rhodamine,thick,domain=15:195] plot ({2*cos(\x)}, {2*sin(\x)});
			\node[Rhodamine] at (-0.7,2.4) {$\bm{S_+{(n)}}$};
			\node[Cyan] at (0.7,-2.4) {$\bm{S_-{(n)}}$};
			\draw[->] (0,0) -- (0.6*2,0.8*2);
			\node at (1.35,1.8) {$\gamma$};
			\draw [black,thick,domain=0:360,fill=white] plot ({2*(0.04*cos(\x)-0.966)}, {2*(0.04*sin(\x)-0.257)});
			\draw [black,thick,domain=0:360,fill=white] plot ({2*(0.04*cos(\x)+0.966)}, {2*(0.04*sin(\x)+0.257)});
			\draw[ ->, domain=15:0] plot ({2.3*cos(\x)}, {2.3*sin(\x)});
			\draw[ ->, domain=-165:-180] plot ({2.3*cos(\x)}, {2.3*sin(\x)});
		\end{tikzpicture}
		\subcaption{}
	\end{subfigure}
	\caption{Target intervals for $\seq{u_n}$ and $\seq{v_n}$ in the proof of \Cref{sign-eap}}
	\label{fig:sign-patterns} 
\end{figure}
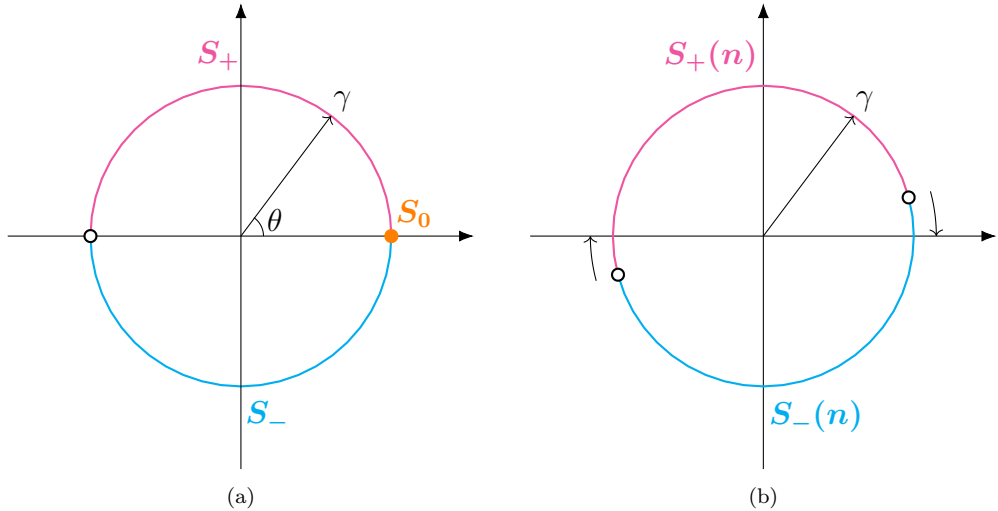

\begin{lemma}
\label{sign-eap}
Both $\alpha$ and $\beta$ are effectively almost-periodic.
\end{lemma}

%We observe the following as a corollary of the proof of Lemma \ref{sign-eap}.
%\begin{corollary}
%\label{sign-ap}
%The sign description of any non-degenerate LRS $u(n)$ of the form $v(n) + r(n)$, where $v(n) = \rho^n n^k \sin(n\theta - \varphi)$, and $r(n) \in o(\rho^n n^k)$, is almost-periodic.
%\end{corollary}

\begin{proof}
	Sequences $\seq{u_n}$ and $\seq{v_n}$ are non-degenerate LRS of order 2 and~4, respectively.
	Hence by \cite{berstel-deux-decidable-properties-of-lrs} both sequences have finitely many zeros.
	In fact, we can identify all of them.
	%using the algorithm for deciding the Skolem problem for sequences of order 4.
	Our sequences satisfy recurrence relations $u_{n+2} = 1.2 u_{n+1} + u_n$ and
	\[
	v_{n+4} = 2.4v_{n+3} - 3.44 v_{n+2} + 2.4v_{n+1} - v_n.
	\]
	Since $\gamma$ is not a root of unity, it is immediate that $u_n = 0$ only for $n = 0$.
	We can determine all zeros of $\seq{v_n}$ either using the general algorithm for solving the Skolem problem for LRS over $\ralg$ of order four \cite{mignotte-shorey-tijdeman-skolem,vereschagin-skolem}, or a simple approach based on the \emph{(absolute logarithmic) Weil height}.
	The Weil height $h(z)$ of an algebraic number has the following properties:
	\begin{itemize}
		\item[(a)] $h(z) > 0$ for every non-zero $z$ that is not a root of unity;
		\item[(b)] $h(k) = \Log|k|$ for $k \in \intg\setminus\{0\}$;
		\item[(c)] $h(z^n) = n h(z)$ for every $z \in \alg$ and $n \in \intg$; 
		\item[(d)] $h(z\cdot y), h(z+y) \le h(z) + h(y) + \Log(2)$ for all $z,y \in \alg$.
	\end{itemize}
	See \cite{waldschmidt2000} for a detailed discussion of the Weil height.
	We have that $v_n = 0$ if and only if $z^n = y_n$, where $z = \gamma/\overline{\gamma}$ and $y_n = \frac{7-n\im}{7+n\im}$.
	Both $z$ and $y_n$ for all $n$ are algebraic numbers of degree at most 2.
	From (c) and (d), $h(y_n) < C\Log n$ for an effectively computable constant $C$, whereas $h(z^n) = nh(z)$ by (b).
	Since $\gamma$ is non-zero and not a root of unity, $h(\gamma) \ne 0$.
	Therefore, $h(z^n)$ grows linearly, whereas $h(y_n)$ grows logarithmically in $n$.
	Equating $h(z^n)$ to $h(y_n)$, we conclude that $v_n \ne 0$ for all $n \ge N$, where $N$ is effectively computable.
	Checking all $n \le N$ individually, we find that for $n \ge 1$, $v_n \ne 0$.
	Therefore, $z(n), y(n) \in \{+,-\}$ for all $n \ge 1$.
	
	\Cref{fig:sign-patterns}~(a) describes how $\alpha \in \Tcal_{\mathit{SA}}$ is generated.
	Both $S_+$ and $S_-$ are open subsets of $\torus$, and $S_0 = \{1\}$.
	For all $n \in\nat$, $\alpha(n)$ is $+$ if and only if $\gamma^n \in S_+$ and $\alpha(n)$ is $-$ if and only if $\gamma^n \in S_-$.
	Since $\alpha(n) \in \{+,-\}$ for $n \ge 1$, $\alpha[1,\infty)$ is generated by $(\gamma, \{\gamma^{-1}S_+,\gamma^{-1}S_-\})$.
	Applying \Cref{toric-TO-TSA-seap}, $\alpha[1,\infty)$ and hence $\alpha$ are both effectively almost-periodic.

	Let us consider $\beta$ next.
	Let $\delta_n = {\arctan (7/n) \in (0,\pi/2)}$, $S_+{(n)} = {e^{\im\delta_n} S_+}$, and $S_-{(n)} = {e^{\im\delta_n}S_-}$.
	We have that for $n \ge 1$, $v_n > 0$ if and only if $\gamma^n \in S_+{(n)}$ and $v_n < 0$ if and only if $\gamma^n  \in S_-{(n)}$.
	\Cref{fig:sign-patterns}~(b) depicts $S_+{(n)}$ and $S_-{(n)}$ for $n = 30$.
	Since $\seq{e^{-\im\delta_n}}$ converges to $1$, as $n \to \infty$, $S_+{(n)}$ uniformly approaches the upper half $S_+$ of the unit circle, whereas $S_-{(n)}$ approaches~$S_-$.
	
	To prove effective almost periodicity of $\beta$, consider a finite word
	\[
	w = w(0)\cdots w(l-1) \in \{+,-\}^l.
	\]
	This word occurs at position $n \ge 1$ in $\beta$ if and only if
	\[
	\bigwedge_{j=0}^{l-1} \gamma^{n+j} \in S_{w(j)}{(n+j)} \:\Longleftrightarrow\: \gamma^n \in \bigcap_{j=0}^{l-1}\gamma^{-j} S_{w(j)}{(n+j)}.
	\]
	Define $S_w(n) = \bigcap_{j=0}^{l-1}\gamma^{-j} S_{w(j)}{(n+j)}$.
	We will argue that either $w$ occurs finitely often in $\beta$, or there exists an open interval subset $K$ of $\torus$ such that $K \subset S_w(n)$ for all sufficiently large $n$.
	
	Recall that for distinct $z_1,z_2 \in \torus_\Gamma$, $\Ical(z_1, z_2)$ is the open interval subset of $\torus$ with endpoints $z_1$ and $z_2$, generated by rotating counter-clockwise starting at $z_1$.
	Each $\gamma^{-j}S_{w(j)}({n+j})$ is of the form $e^{\im\delta_{n+j}}\gamma^{-j}I_j$, where $I_j$  is  $S_+$ if $w(j)$ is the letter $+$ and $I_j = S_-$ otherwise.
	Since $\delta_n = \Theta(1/n)$, 
	$\gamma^{-j}S_{w(j)}({n+j})$ uniformly approaches the interval $\gamma^{-j}I_j$ as $n \to \infty$.
	
	The endpoints of $\gamma^{-j}I_j$ are $\gamma^{-j}$ and $-\gamma^{-j}$.
	As $\gamma$ is not a root of unity, for every $j_1 \ne j_2$, $\gamma^{-j_1}$ is not equal to $\gamma^{-j_2}$ and $-\gamma^{-j_2}$.
	Hence the limit intervals $\gamma^{-j}I_j$ for $0 \le j < l$ have $2l$ distinct endpoints in total.
	Therefore,
	\begin{itemize}
		\item[(a)] either there exists $N$ such that $S_w(n)$ is empty for all $n \ge N$ (which happens if and only if the ``limit shape'' $\bigcap_{j=0}^{l-1} \gamma^{-j}I_j$ is empty), or
		\item[(b)] there exists $N$ such that for all $n \ge N$, $S_w(n) = \Ical(z_1 e^{\im \delta^{(1)}_n}, z_2 e^{\im \delta^{(2)}_n})$ is non-empty, where $z_1,z_2$ are distinct and of the form $\pm \gamma^{-j}$ for some $0 \le j <l$ and $\delta^{(1)}_n,\delta^{(2)}_n = \Theta(1/n)$.
	\end{itemize} 
	Since all steps above are effective, we can effectively compute $N$ in both cases, and in case~(b), construct a $\rat$-semialgebraic interval $J$
	%$J = \Ical(y_1,y_2)$ with $y_1,y_2 \in \torus \cap \alg$ 
	such that for all $n \ge N$, $J \subset S_w(n)$.
	In case~(a) the word $w$ does not occur in $\beta[N,\infty)$ and we are done. 
	Otherwise, observe that for $n \ge N$, $\gamma^n \in J \Rightarrow \beta[n,n+l) = w$.
	Since the endpoints  of $J$ are algebraic, we can compute $K$ such that for all $m \in \nat$, $\gamma^n \in J$ for some $m \le n \le m+K$; see the proof of \Cref{toric-TO-strongly-AP} for the usual topological construction, or \cite[Lemma~2]{karimov-ltl} for a direct formula.
	We conclude that the word~$w$ occurs in every subword of $\beta$ of length $N + K + l$.
\end{proof}

The discussion above suggests to think of $\beta$  as being ``toric with moving targets''.
We next show that $\alpha \times \beta$ is radically different from both $\alpha$ and $\beta$, and far from belonging to $\Tcal_{O}$ or $\Tcal_{\mathit{SA}}$.
%Intuitively, the reason is that we can apply a renaming of letters (\ie a 1-uniform morphism) to $\alpha \times \beta$ to obtain a word $\kappa \in \{0,1\}$ such that for all~$n$, $\kappa(n) = 1$ if and only if $\gamma^n \in I_n$, where $\seq{I_n}$ is a sequence of intervals in $\torus$ whose extent at time~$n$ is $\Theta(1/n)$.
%The frequency of occurrence of the letter 1 is related to the measure of $I_n$ in $\torus$, and since the latter vanishes as $n \to \infty$

\begin{theorem}
	For $\alpha, \beta$ as in \Cref{sign-eap}, the word $\alpha \times \beta$ is not almost-periodic and hence does not belong to $\Tcal_O \cup \Tcal_{\mathit{SA}}$.
\end{theorem}
\begin{proof}
	Recall from \Cref{toric-TO-strongly-AP,toric-TSA-AP} that all words belonging to $\Tcal_O$ or $\Tcal_{\mathit{SA}}$ are almost-periodic. 
	We therefore only need to prove the first statement.
	We will show that (a) the letter $(+,-)$ occurs infinitely often in $\alpha \times \beta$, and (b) the length of the gaps between its consecutive occurrences is not bounded. 
	
We start with (a). 
The letter $(+, -)$ occurs at a position $n > 0$ if and only if $\sin (n\theta) > 0$ and $n\sin (n\theta) - 7\cos (n\theta) < 0$, which is equivalent to  $0 <  \Log(\gamma^n) < \arctan (7/n)$.
We will show that $0 < \Log (\gamma^n) < 2\pi/n$ is satisfied for infinitely many $n \in \nat$. 
Since $\arctan (7/n) > 2\pi/n$ for $n > 11$, this proves that $(+,-)$ occurs infinitely often in $\alpha \times \beta$.

Let $t = \Log(\gamma)/(2\pi\im) \in (0,1)\setminus \rat$.
For $n \ge 1$, $\Log(\gamma^n) \in (0, 2\pi/n)$ if and only if $nt - \lfloor nt \rfloor < 1/n$.
We find infinitely many values of $n$ satisfying the latter inequality using the \emph{continued fraction expansion} of $t$:
$$
t = \cfrac{1}{a_1 + \cfrac{1}{a_2 + \cfrac{1}{a_3 + \ddots}}}
$$
where each $a_i$ is a positive integer; see \cite{borwein_van-der-poorten_shallit_zudilin_2014}.
Let $p_n/q_n$ be \emph{$n$th convergent}. That is, $p_n/q_n$ is the rational approximation of $t$ obtained by truncating the expansion at the $n$th level. 
For all $n$, we have that
\[
q_{n+1}t -p_{n+1} = \frac{(-1)^{n+1}}{a_{n+2}q_{n+1}+q_n}.
\]
In particular, the $n$th convergent is an over-approximation when $n$ is odd and an under-approximation when $n$ is even.
Moreover, $|p_n/q_n-t| < 1/q_n^2$ for all $n$, and $\seq{q_n}$ is strictly increasing.
Therefore, for every even $n \ge 1$, 
\[
0 < t -  \frac{p_n}{q_n} < \frac{1}{q_n^2}
\]
and hence $q_nt - \lfloor q_nt \rfloor < 1/q_n$.

We move on to proving (b).
Let $J_n = S_+ \cap S_-(n)$.
Recall that the letter $(+,-)$ occurs at the position $n$ in $\beta$ if and only if $\gamma^n \in J_n$, and the length of~$J_n$ is $\Theta(1/n)$.
Let $B \in \nat$.
We show how to construct $n$ such that letter $(+,-)$ does not occur in $\beta[n,n+B)$.
Let $m$ be sufficiently large that $\torus \setminus \bigcup_{i=0}^{B} \gamma^{-i}J_{m}$ contains a non-empty open subset $O$ of  $\torus$.
Further let $n \ge m$ be such that $\gamma^n \in O$.
By construction, for every $0 \le i \le B$, $\gamma^{n+i} \notin J_{m}$.
Since $J_{m+i} \subset J_{m}$ for all $i \in \nat$, we have that for all $0 \le i < B$, $\gamma^{n+i} \notin J_{n+i}$.
That is, for all $0 \le i < B$, $(\alpha \times \beta)(n+i)$ is not the letter $(+,-)$.
\end{proof}
\begin{corollary}
	The word $\beta$ does not belong to $\Tcal_{O} \cup \Tcal_{\mathit{SA}}$. 
\end{corollary}
\begin{proof}
	%Since $\alpha \times \beta$ is not almost-periodic, it does not belong to $\Tcal_{O} \cup \Tcal_{\mathit{SA}}$.
	Recall that $\alpha$ belongs to both $\Tcal_{O}$ and $\Tcal_{\mathit{SA}}$, and both classes are closed under products.
	Since $\alpha \times \beta$ does not belong to $\Tcal_{O} \cup \Tcal_{\mathit{SA}}$, neither does~$\beta$.
\end{proof}

We mention that \cite[Ex.~2]{salimov2010uniform} gives an example of two uniformly recurrent morphic words whose product is not almost-periodic.

\subsection{Characteristic words of linear dynamical systems}
\label{sec-lds}

One application of toric words and  MSO decidability that has recently received significant attention is the \emph{model-checking problem (MCP)} for linear dynamical systems (LDS) \cite{Karimov2022}.
An LDS is given by a pair $(M,s)$ where ${M \in \rat^{d \times d}}$ is the update matrix and $s \in \rat^d$ is the starting configuration.
The \emph{orbit} of $(M,s)$ is the infinite sequence $\seq{M^ns}$.
Let $\Scal = \{S_1,\ldots, S_m\}$ be a collection of $\rat$-semialgebraic subsets of $\rel^d$.
Writing $\Sigma=2^\Scal$, the \emph{characteristic word} of $(M,s)$ with respect to $\Scal$ is the word $\alpha \in \Sigma^\omega$ defined by $S_i \in \alpha(n)\Longleftrightarrow M^ns \in S_i$ for all $1 \le i \le m$ and $n \in \nat$.
The model-checking problem is to decide, given $(M,s)$ and a deterministic automaton $\aut$, whether $\aut$ accepts $\alpha$.
If we fix $M,s, \Scal$, and only let $\aut$ vary, by B\"uchi's result \cite{buchi-collected-works}, the resulting problem is Turing-equivalent to the decision problem for the MSO theory of $\langle\nat; <, P_1,\ldots, P_m \rangle$, where each $P_i\st \nat \to \{0,1\}$ is the binary predicate defined by $P_i(n) = 1$ if and only if $M^ns \in S_i$ for all $n \in \nat$.

Let $p_1,\ldots, p_K$ be all polynomials (with rational coefficients) appearing in the definition of $\Scal$.
For each $1\le j \le K$, the sequence $u_n = p_j(M^ns)$ is an LRS over $\rat$.
Denote its sign pattern by $\alpha_j \in \{+,0,-\}^\omega$.
Since each $S_i$ is generated by a Boolean combination of polynomial inequalities, we have that $\alpha = \sigma(\alpha_1\times \cdots \times \alpha_K)$, where $\sigma$ is a 1-uniform morphism.
Hence understanding the characteristic word of an LDS with respect to a collection of semialgebraic sets $\Scal$ boils down to understanding sign patterns of a collection of linear recurrence sequences.

The model-checking problem for LDS subsumes, among many others, the Skolem problem, the Positivity problem, and the Ultimate Positivity problem for LRS over $\rat$.
Unsurprisingly, decidability of the full model-checking problem is currently open.
However, decidability can be proven if we place certain restrictions on $M, \aut$, and $\Scal$.
\begin{enumerate}
	\item[(A)] Call a $\rat$-semialgebraic set $T$ \emph{low-dimensional} if it either has intrinsic (\ie semialgebraic) dimension 1, or is contained in a three-dimensional linear subspace.
	The set $T$ is \emph{tame} if it can be obtained in finitely many steps from a collection of low-dimensional sets through the usual set operations.
	If all targets in~$\Scal$ are tame, then the characteristic word $\alpha$ of any LDS $(M,s)$ with respect to $\Scal$ is \eap \cite{Karimov2022,LinearLoopsPOPL}.
	In particular, $\alpha$ has a suffix belonging to the class $\Tcal_{O}\cap\Tcal_{\mathit{SA}(\rat)}$ that is fully effective.
	Hence the MCP with tame targets (but arbitrary $(M,s)$ and $\aut$) is decidable.
	\item[(B)] An automaton $\aut$ is \emph{prefix-independent} if for every infinite word $\beta$, whether $\aut$ accepts $\beta$ does not change if we perform finitely many insertions and deletions on $\beta$.
	It is shown in \cite{karimov-popl21} that the MCP is decidable if we assume $M$ is diagonalisable and $\aut$ is prefix-independent.
\end{enumerate}
From (A) it follows that the MCP is decidable in dimension at most~3.
On the other hand, (B) is closely related to \Cref{simple-lrs-sign}.
To see this, suppose $M$ is diagonalisable.
Then $u_n = p(M^ns)$ is a simple LRS for every polynomial $p$.
From the connection between the characteristic word $\alpha$ and the sign patterns of LRS defining $\Scal$ discussed above, the closure properties of toric words, as well as \Cref{simple-lrs-sign}~(a), it follows that $\alpha$ has a suffix that belongs to $\Tcal_{O} \cap \Tcal_{\mathit{SA}}$.
Unfortunately, it is not known how to determine the starting position of such a suffix in $\alpha$, which is the reason why in (B) we impose the prefix-independence restriction.
However, similarly to \Cref{simple-lrs-sign}~(b), it is shown in \cite{karimov-power-of-positivity} that the MCP is decidable for diagonalisable LDS if we assume decidability of the Positivity problem for simple LRS over $\rat$.



\bibliographystyle{elsarticle-num} 
\bibliography{refs.bib}

\end{document}